\documentclass[11pt,letter]{article}

\usepackage[small]{complexity}
\usepackage{fullpage}
\usepackage{amssymb,amsmath,amsfonts,amsthm,latexsym,graphicx}
\usepackage{hyperref}
\usepackage[nothing]{algorithm}
\usepackage{algpseudocode}
\usepackage[shortlabels]{enumitem}
\usepackage[normalem]{ulem}
\usepackage[font=small,labelfont=bf]{caption}
\usepackage{lmodern}

\usepackage[dvipsnames]{xcolor}
\definecolor{LinksColor}{rgb}{0,0.2,0.8}
\hypersetup{colorlinks,linkcolor=LinksColor,urlcolor=LinksColor,citecolor=LinksColor}
\newcommand{\lref}[2][]{\hyperref[#2]{#1~\ref*{#2}}}

\newcommand{\myparagraph}[1]{{\medskip\noindent{\bf #1}}}
\newcommand{\emmedparagraph}[1]{{\medskip\noindent\emph{#1}}}
\newcommand{\etal}{{\em et al.}}

\newcommand{\mycase}[1]{\mbox{{\underline{Case #1}}:\/}}

\newcommand{\JG}{{\cal J}_G}

\newcommand{\mbfx}{\boldsymbol{x}}

\newcommand{\calJ}{{\cal J}}

\newcommand{\calI}{{\cal I}}

\renewcommand{\E}{\mathbf{E}}

\newcommand{\hatmbfx}{\hat{\boldsymbol{x}}}

\newcommand{\vertbar}{{\Big|}}

\newcommand{\braced}[1]{{ \left\{ #1 \right\} }}

\newcommand{\suchthat}{{\;:\;}}

\newcommand{\Exp}{\mbox{\bf{E}}}
\newcommand{\Prob}{\textrm{Pr}}

\newcommand{\floor}[1]{{\lfloor #1\rfloor}}

\newcommand{\nat}{\mathbb{N}}
\newcommand{\rational}{\mathbb{Q}}
\newcommand{\reals}{{\mathbb R}}

\newcommand{\JRP}{{\mbox{JRP}}}
\newcommand{\JRPD}{{\mbox{JRP-D}}}

\newcommand{\Efour}{{\text{E4}}}
\newcommand{\JRPDEfour}{\JRPD_{\Efour}}

\newcommand{\VG}{{\mbox{VG}}}
\newcommand{\EG}{{\mbox{EG}}}
\newcommand{\SG}{{\mbox{SG}}}

\newcommand{\llQ}{Q^0}
\newcommand{\lQ}{Q^1}
\newcommand{\rQ}{Q^2}
\newcommand{\rrQ}{Q^3}

\newtheorem{theorem}{Theorem}

\newtheorem{lemma}{Lemma}

\newcommand{\OPT}{\textsc{Opt}}

\newcommand{\half}{{\textstyle\frac{1}{2}}}
\newcommand{\onehalf}{{\textstyle\frac{1}{2}}}
\newcommand{\threehalves}{{\textstyle\frac{3}{2}}}
\newcommand{\onethird}{{\textstyle\frac{1}{3}}}
\newcommand{\twothirds}{{\textstyle\frac{2}{3}}}

\newcommand{\fivesixths}{{\textstyle\frac{5}{6}}}

\newcommand{\demands}{{\cal D}}
\newcommand{\COST}{\textsf{C}}
\newcommand{\cost}{\textsf{c}}
\newcommand{\costVector}{\cost}
\newcommand{\costof}[1]{\cost_{#1}}
\newcommand{\costfn}[1]{\textsf{cost}(#1)}

\newcommand{\universe}{{\cal U}}  
\newcommand{\maxDeadline}{{\mbox{\small $U$}}}  
\renewcommand{\maxDeadline}{U}  
\newcommand{\duration}[1]{\delta_{#1}}
\newcommand{\confDeadline}{\Delta}  

\newcommand{\distribution}{p} 
\newcommand{\statistic}[1]{{\cal Z}(#1)} 
\newcommand{\EWaste}[2]{{\cal W}(#1,#2)}  
\newcommand{\threshold}{z} 

\newcommand{\round}{\mbox{\sf Round}}  
\newcommand{\sample}[1]{s_{#1}}   
\newcommand{\state}[1]{\psi_{#1}} 

\newcommand{\GOTstop}{I} 

\newcommand{\tmpRHO}{\rho}

\newcommand{\AREAFN}{\omega_{\tmpRHO}}  
\renewcommand{\AREAFN}{\tau_{\!\tmpRHO}}  
\newcommand{\AREA}[1]{\AREAFN({#1})}  

\newcommand{\MEASFN}{\sigma}  
\newcommand{\MEAS}[1]{\MEASFN({#1})}  

\newcommand{\LOTstop}{J_{\tmpRHO}} 

\newcommand{\adjustedcost}{{\textit{AC}}} 

\floatname{algorithm}{Pseudocode}

\begin{document}

\title{Approximation Algorithms for \\ the Joint Replenishment Problem 
  with Deadlines\thanks{%
	The final publication appeared in Journal of Scheduling and is available at Springer via 
	\url{http://dx.doi.org/10.1007/s10951-014-0392-y}.
    Research supported by NSF grants CCF-1217314, CCF-1117954, OISE-1157129;
    EPSRC grants EP/J021814/1 and EP/D063191/1;
    FP7 Marie Curie Career Integration Grant;
    Royal Society Wolfson Research Merit Award; and 
	Polish National Science Centre grant DEC-2013/09/B/ST6/01538.
 }
}

\author{Marcin Bienkowski\thanks{%
			Institute of Computer Science, 
			University of Wroc{\l}aw, Poland.
			}
		\and
		Jaros{\l}aw Byrka\footnotemark[2]
		\and
		Marek Chrobak\thanks{%
			Department of Computer Science,
         	University of California at Riverside, USA.
			}
		\and
		 Neil Dobbs\thanks{%
				IBM T.J. Watson Research Center, Yorktown Heights, USA.
			}
		\and
		Tomasz Nowicki\footnotemark[4]		
		\and
		Maxim Sviridenko\thanks{%
			Department of Computer Science,
			University of Warwick, UK.
			}
		\and
		Grzegorz {\'S}wirszcz\footnotemark[4]
		\and
		Neal E.~Young\footnotemark[3]
}

\maketitle


\begin{abstract}
The Joint Replenishment Problem ($\JRP$) is a fundamental optimization problem in 
supply-chain management, concerned with optimizing the flow of goods
from a supplier to retailers. Over time, in response to demands at the retailers,
the supplier ships orders, via a warehouse, to the retailers.   
The objective is to schedule these orders to minimize the sum of ordering costs and 
retailers' waiting costs.

We study the approximability of $\JRPD$, the version of $\JRP$ with deadlines,
where instead of waiting costs the retailers impose strict deadlines. We study
the integrality gap of the standard linear-program (LP) relaxation, giving a
lower bound of $1.207$, a stronger, computer-assisted lower bound of $1.245$,
as well as an upper bound and approximation ratio of $1.574$. The best
previous upper bound and approximation ratio was $1.667$; no lower bound was
previously published. For the special case when all demand periods are of
equal length we give an upper bound of $1.5$, a lower bound of $1.2$, and show
APX-hardness.
\end{abstract}


\section{Introduction}
\label{sec: introduction}


The \emph{Joint Replenishment Problem with Deadlines} ($\JRPD$) is an optimization 
problem in supply-chain management concerned with scheduling shipments
(orders) of a commodity from a supplier, via a shared warehouse, 
to satisfy prior demands at $m$ retailers (cf.~\lref[Figure]{fig:problem}).
The objective is to find a schedule of orders that satisfies all
demands before their deadlines expire, while minimizing the total ordering cost.

Specifically, an instance of $\JRPD$ is given by a tuple $(\COST, \costVector, \demands)$
where 
\begin{itemize}
	\item $\COST \in \rational$ is the \emph{warehouse ordering cost};
	\item $\costVector$ is the vector of \emph{retailer ordering costs}, where
		for each retailer $\rho\in\{1,2,\ldots,m\}$ its ordering cost is
		$\costof{\rho} \in \rational$;
	\item $\demands$ is a set of $n$ \emph{demands}, with each demand represented
	 	by a triple $(\rho,r,d)$, 
		where $\rho$ is the retailer that issued the demand, 
		$r\in\rational$ is the demand's release time 
		and $d\in\rational$ is its deadline.  
\end{itemize}
For a demand $(\rho,r,d)$,
the interval $[r,d]$ is called the demand \emph{period}\footnote{Note: our use of the term
``period'' is different from its use in operations research literature on supply-chain
management problems.}.  
In sections that prove upper bounds
we assume (without loss of generality by time scaling) that $r,d\in [2n]$,
where $[i]$ denotes $\{1,2,\ldots,i\}$.

A solution (also called a \emph{schedule}) is a set of \emph{orders},
each specified by a pair $(t,R)$, where $t$ is the time of the order 
and $R$ is a subset of the retailers.
An order $(t,R)$ \emph{satisfies} those demands $(\rho,r,d)$
whose retailer is in $R$ and whose demand period contains $t$
(that is, $\rho \in R$ and $t\in [r, d]$).
A schedule is \emph{feasible} if all demands are satisfied by some order
in the schedule.

The \emph{cost} of order $(t,R)$ is 
the ordering cost of the warehouse plus the ordering costs of respective
retailers, i.e., $\COST + \sum_{\rho \in R} \costof{\rho}$. 
It is convenient to think of
this order as consisting of a warehouse order of cost $\COST$, which is then joined
by each retailer $\rho\in R$ at cost $\costof{\rho}$. 
The cost of the schedule is the sum of the costs of its orders.
The objective is to find a feasible schedule of minimum cost.


\begin{figure}[t]
\centering
\includegraphics[scale=0.75]{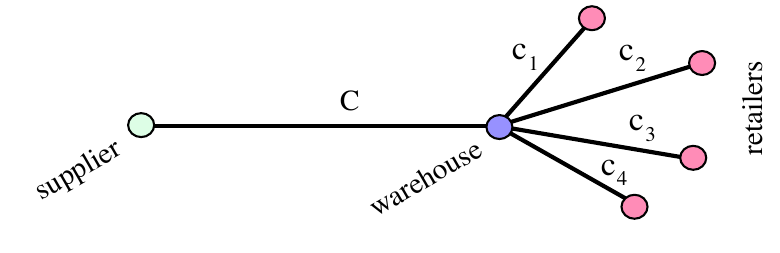}
\caption{An instance with four retailers, represented by a tree with ordering costs as weights assigned to its edges. 
  The cost of an order is the total weight of the subtree connecting the supplier and the involved retailers.
} 
\label{fig:problem}
\end{figure}


\myparagraph{Previous results.}
The decision variant of {\JRPD} was shown to be strongly $\NP$-complete
by Becchetti~{\etal}~\cite{packet-aggregation-becchetti}.  
(They considered an equivalent problem of packet aggregation with deadlines on two-level trees.) 
Nonner and Souza~\cite{jrp-deadlines-nonner} then showed that
{\JRPD} is $\APX$-hard, even if each retailer issues only three demands.
Levi, Roundy and Shmoys~\cite{jrp-levi-2-approx} gave a \mbox{$2$-approximation} algorithm based on a primal-dual scheme.
Using randomized rounding,
Levi~{\etal}~\cite{jrp-owmr-levi-journal,jrp-owmr-levi-approx}
(building on \cite{jrp-owmr-levi-soda})
improved the approximation ratio to 1.8;
Nonner and Souza~\cite{jrp-deadlines-nonner}
reduced it further to $5/3$.
These results use a natural linear-program (LP) relaxation, which we use too.

The randomized-rounding approach from~\cite{jrp-deadlines-nonner}
uses a natural rounding scheme whose analysis 
can be reduced to a probabilistic game.
For any probability distribution $\distribution$ on $[0,1]$,
the integrality gap of the LP relaxation is at most $1/\statistic \distribution$,
where $\statistic \distribution$ is a particular statistic of $\distribution$
(see \lref[Lemma]{lemma: round}).
The challenge in this approach is to find a distribution where $1/\statistic \distribution$ is small.
Nonner and Souza show that there is a distribution $\distribution$ with 
$1/\statistic \distribution \le 5/3 \approx 1.67$.
As long as the distribution can be sampled from efficiently,
the approach yields a polynomial-time $(1/\statistic \distribution)$-approximation algorithm.


\myparagraph{Our contributions.}
We prove that there is a distribution $\distribution$ with $1/\statistic \distribution\le 1.574$.
We present this result in two steps: 
we show the bound $e/(e-1) \approx 1.58$ with a simple and elegant analysis, 
then improve it to $1.574$ by refining the underlying distribution.
This shows that the integrality gap is at most $1.574$
and it gives a $1.574$-approximation algorithm.
We also prove that the LP integrality gap is {\em at least} $1.207$ and
we provide a computer-assisted proof that this gap is at least $1.245$.
As far as we know, no explicit lower bounds have been previously published.

For the special case when all demand periods have the same length 
(as occurs in applications where time-to-delivery is globally standardized)
we give an upper bound of 1.5, a lower bound of 1.2, 
and show $\APX$-hardness.


\myparagraph{Other related work.}
{\JRPD} is a special case of the Joint Replenishment Problem ({\JRP}).  
In {\JRP}, instead of having a deadline, each demand is associated with a delay-cost
function that specifies the cost for the delay between the time the demand is
released and the time it is satisfied by an~order. {\JRP} is $\NP$-complete, even if the 
delay cost is linear~\cite{jrp-arkin,jrp-deadlines-nonner}.  {\JRP} is in
turn a special case of the One-Warehouse Multi-Retailer (OWMR) problem, 
where the commodities may be stored at the warehouse for a~given cost per time unit.
The $1.8$-approximation by Levi~{\etal}~\cite{jrp-owmr-levi-approx} holds also for OWMR. 
{\JRP} was also studied in the online scenario: a $3$-competitive algorithm was given
by Buchbinder~{\etal}~\cite{jrp-online-buchbinder}  (see also \cite{aggregation-bkv}).

The $\JRP$ model is an abstraction of a number of other optimization problems 
that arise in supply-chain management. It is often presented as an inventory-management 
problem, where all demands need to be satisfied immediately from the
current inventory. In that scenario, orders are issued to replenish the
inventory, ensuring that all future demands are met. (In contrast, in our model 
the orders are issued to satisfy \emph{past} demands and there is no inventory.)
Depending on the
application, orders can represent deliveries (via a shared warehouse), or a
manufacturing process that involves a joint set-up cost and 
individual set-up costs for retailers. The objective is to
minimize the total cost, defined as the sum of ordering costs and inventory
holding costs. 

Another generalization of {\JRP} involves a tree-like structure with
the supplier in the root and retailers at the leaves,
modeling control packet aggregation in computer networks.
A $2$-approximation is known for the variant with deadlines~\cite{packet-aggregation-becchetti};
the case of linear delay costs has also been studied~\cite{khanna-message-aggregation,aggregation-bkv}.
Recently, L.~Chaves (private communication) has shown that
the generalization of {\JRP} to arbitrary trees, even for arbitrary waiting cost functions, 
can be approximated within a factor of $2$ through a reduction to the
multi-stage assembly problem, see~\cite{jrp-levi-2-approx}.



\section{Upper Bound of 1.574}
\label{sec: upper bounds}



In this section we derive our approximation algorithms for $\JRPD$, showing an
approximation ratio of $e/(e-1)\approx 1.58$, which we then improve to $1.574$.
Both algorithms are based on randomized LP-rounding.


\myparagraph{The LP relaxation.}
For the rest of this section, fix an arbitrary instance
$\calI = (\COST, \costVector,\demands)$ of {\JRPD}.
Let finite set $\universe\subset\rational$ contain the release times and deadlines.
Here is the standard LP relaxation of the problem:

\begin{alignat}{3}
	\textrm{minimize}\quad\quad  \costfn{\mbfx} \;&=\;  
        \textstyle \sum_{t\in\universe} {(\COST\,x_{t} + \sum_{\rho=1}^m{\costof{\rho}\, x^\rho_t})}
											\hspace{-1.3in}&&
	\notag \\
	\textrm{subject to}\quad\quad\quad\quad 
			\textstyle x_t, x^\rho_t \;&\ge 0\;
											&	&\textrm{for all}\; t \in \universe, \rho \in \{1, \ldots, m\}
	\notag
        \\
        \textstyle	x_t  \;&\geq\;  x^\rho_t 
	 										&	&\textrm{for all}\; t \in \universe, \rho \in \{1, \ldots, m\}
	\label{eqn: density bound}
	\\
			 \sum_{t\in \universe\cap[r,d]} {x^\rho_t} \;&\geq 1\;
											&	&\textrm{for all}\; (\rho,r,d) \in \demands.
	\label{eqn: satisfy bound}
\end{alignat}

\myparagraph{The statistic $\statistic \distribution$.}
Let $\distribution$ be a probability distribution on $[0,1]$. As we are about to show,
the approximation ratio of algorithm $\round_\distribution$ (defined below)
and the integrality gap of the LP
are at most $1/\statistic \distribution$, where $\statistic \distribution$ 
is defined by the following so-called {\em tally game} (following \cite{jrp-deadlines-nonner}).
To begin the game, fix any {\em threshold} $\threshold \ge 0$,
then draw a sequence of independent samples $\sample 1,\sample 2,\ldots,\sample h$ 
from $\distribution$, stopping when their sum exceeds $\threshold$,
that is when $\sample 1+\sample 2+\ldots+\sample {h} > \threshold$.
Call $\threshold- (\sample 1+\sample 2+\ldots+\sample {h-1})$ the {\em waste}.
Note that, since the waste is less than $\sample h$, it is in $[0,1)$.
Let $\EWaste \distribution \threshold$ denote the expectation of the waste.
Abusing notation, let $\Exp[\distribution]$ denote the expected value of a~single sample drawn from~$\distribution$.
Then $\statistic \distribution$ is defined by

\begin{equation*}
\statistic \distribution = \min\Big\{\, \Exp[\distribution]\,,\, 
		1-\sup_{\threshold \ge 0} \EWaste \distribution \threshold \,\Big\}.
\end{equation*}

Note that the distribution $\distribution$ that chooses $\half$ with probability 1 has 
$\statistic \distribution = \half$.
The challenge is to simultaneously increase $\E[p]$ and reduce the maximum expected waste.


\myparagraph{A generic randomized-rounding algorithm.}
The upper bound of 1.574 relies on a randomized-rounding algorithm, $\round_\distribution$.
The algorithm is parameterized by an arbitrary probability distribution $\distribution$ on $[0,1]$
and gives a $(1/\statistic \distribution)$-approximation:

\begin{lemma}\label{lemma: round}
For any distribution $\distribution$ on $[0,1]$ and fractional LP solution $\mbfx$,
if $\statistic p > 0$,
then with probability 1, Algorithm~$\round_\distribution(\COST,\costVector,\demands,\mbfx)$ 
returns a feasible schedule.
The expected cost of the schedule is at most $\costfn{\mbfx}/\statistic \distribution$.
\end{lemma}

The main ideas underlying $\round_\distribution$ and its analysis are from \cite{jrp-deadlines-nonner};
the presentation here (\lref[Section]{sec: details}) addresses some technical subtleties.
Subsequent sections (with a minor exception) use \lref[Lemma]{lemma: round} as a black box;
they can be read independently of the proof in \lref[Section]{sec: details}.


\subsection{The details of \texorpdfstring{$\round_\distribution$}{Round(p)} and 
proof of \texorpdfstring{Lemma~\ref{lemma: round}}{Lemma 1}}\label{sec: details}

Fix an arbitrary optimal fractional solution $\mbfx$ of the LP relaxation for instance $\calI$.
As previously discussed, without loss of generality we can assume that
the given universe $\universe$ of release times and deadlines is $[\maxDeadline]$,
where $\maxDeadline$ ($\le 2n$) is the maximum deadline.
We focus on producing a ``rounded'' schedule $S$ for $\calI$
with expected cost at most $1/\statistic p$ times
$\costfn \mbfx = \sum_{j=1}^\maxDeadline \big(\COST\, x_j + \sum_\rho \costof{\rho}\, x_j^\rho\big)$.

\newcommand{\rmo}{r\mbox{--}1}

\myparagraph{Extend to continuous time.}
To start, we recast the problem of rounding $\mbfx$ as a continuous-time problem.
Extend the universe $\universe$ of times
from the discrete set $\universe = [\maxDeadline]$
to the continuous interval $\overline\universe = (0,\maxDeadline]$
and relax each demand period $[r,d]$, 
replacing it with the half-open period $(\rmo,d]$.
Let $\overline\calI$ denote the modified instance.
To find a schedule $S$ for the given instance $\calI$,
we will find a~schedule $\overline S$ for $\overline\calI$,
then take $S = \{ (\lceil t \rceil, R)\suchthat (t,R) \in \overline S\}$.
$S$ clearly has the cost not larger than that of $\overline S$
and is also feasible (because the release times and deadline are integers).

For the algorithm, reinterpret the fractional solution $\mbfx$ 
as a {\em continuous-time} solution $\overline\mbfx$ over universe $\overline\universe$:
as $t\in\overline\universe$ ranges continuously from $0$ to $\maxDeadline$,
the continuous-time solution 
$\overline\mbfx$ ships continuously at the {\em shipping rate} $\overline x_t = x_{\lceil t\rceil}$
and has each retailer $\rho$ join at his {\em take rate} 
$\overline x^\rho_t = x^\rho_{\lceil t\rceil}$.
The example at the top of \lref[Figure]{fig:bar_chart} illustrates $\overline \mbfx$ over time.


\begin{figure}[t]
  \centering
    \includegraphics[height=1in]{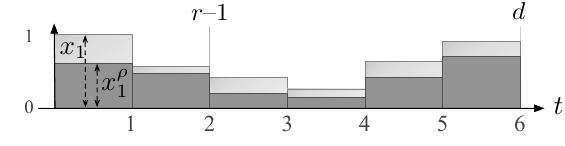}
    \vspace*{-10pt}
    \caption{
      The continuous solution $\overline\mbfx$:
      the universe is $\overline\universe=(0,\maxDeadline]$.
      At each time $t$ in $\overline\universe$,
      the solution ships at rate $\overline x_t= x_{\lceil t\rceil}$
      while each retailer $\rho$ takes at rate $\overline x^\rho_t=x^\rho_{\lceil t \rceil}$.
      During any demand period $(r-1,d]$, the retailer's cumulative take
      is at least 1.
    }
    \label{fig:bar_chart}
\end{figure}


At each time $t$ in $\overline\universe$,
define the {\em total shipped} up to time $t$ to be
$\MEAS{t} = \int_{0}^ t \overline x_t \,dt = \int_0^t x_{\lceil t \rceil}\,dt$.
Define $\rho$'s {\em take} up to time $t$ to be
$\AREA{t} = \int_{0}^ t \overline x^\rho_t \,dt = \int_0^t x^\rho_{\lceil t \rceil}\,dt$.
Over any interval $(t,t']$, define the amount shipped to be $\MEAS{t'}-\MEAS{t}$.
Likewise define $\rho$'s take to be $\AREA{t'}-\AREA{t}$.


\myparagraph{The algorithm.}
$\round_\distribution$ draws samples $(\sample 1, \sample 2, \ldots, \sample \GOTstop)$
i.i.d.~from the distribution $p$, 
stopping when the sum of the samples first exceeds $\MEAS{\maxDeadline}-1$.
It creates orders at the times $\mathbf t = (t_1,t_2,\ldots,t_{\GOTstop})$
such that, for each $i\in[\GOTstop]$,
the continuous-time solution $\overline\mbfx$ 
ships $s_i$ units in interval $(t_{i-1},t_i]$ (interpreting $t_0$ as 0).
By the choice of the number of samples ${\GOTstop}$,
$\overline\mbfx$ ships strictly less than 1 unit in $(t_{\GOTstop},\maxDeadline]$.

After $\round_\distribution$ chooses $\mathbf t$, 
for each retailer $\rho$ independently,
it has $\rho$ join a minimum-size subset of the orders that satisfies $\rho$'s demands.
(It computes this optimal subset using the standard earliest-deadline-first algorithm.)
This determines the schedule $\overline S$ for the continuous-time instance $\overline\calI$.
To get the schedule $S$ for the original instance $\calI$,
the algorithm shifts each order time $t_i$ to its ceiling $\lceil t_i\rceil$.
The algorithm is shown in \lref[Figure]{fig:round}.

We remark that, in practice, modifying the algorithm to round up 
the times as soon as they are chosen might yield lower-cost solutions for some instances.  
That is, take each $t_i$ to be the minimum integer such that the amount shipped over 
interval $(t_{i-1}, t_i]$ is at least $\sample i$ (interpreting $t_0$ as $0$).  
Stop when the amount shipped over $(t_i, U]$ is less than 1.


\newcommand{\StateLong}[1]
{\State\parbox[t]{\dimexpr\linewidth-\algorithmicindent}{#1 \strut}}

\begin{figure}[t]
\centering
\begin{minipage}{\textwidth}
\hrulefill
\begin{algorithmic}[1] \setlength{\parindent}{0in}\setlength{\parskip}{1pt}\setlength{\itemsep}{1pt}

    \StateLong{
      Draw samples $(\sample 1,\sample 2, \sample 3, \ldots, \sample \GOTstop)$
      i.i.d.~from $\distribution$, stopping when the sum of the samples
      first exceeds $\MEAS{\maxDeadline}-1$.
      Schedule orders at times in $\mathbf t = (t_1,t_2,\ldots,t_{\GOTstop})$
      such that, for each $i\in[\GOTstop]$,
      each $t_i$ is the minimum such that
      $\overline\mbfx$ ships $\sample i$ units
      in the interval $(t_{i-1},t_i]$
      (interpreting $t_0$ as~0).
      (By the choice of $\GOTstop$,
      $\overline\mbfx$ ships strictly less than 1 unit
      in interval $(t_{\GOTstop},\maxDeadline]$.)} \label{line:t}

    \For{{\bf each} $\rho\in[m]$} 

    \StateLong{
      Use the earliest-deadline-first algorithm to choose a minimum-size subset
      of the orders for~$\rho$ to join to satisfy his demands.
      More explicitly:
    }\label{line:join}

    \While{retailer $\rho$ has any not-yet-satisfied demand $(\rho,r,d)$}
    \label{line:while}

    \State Let $d^*$ be the earliest deadline of such a demand.

    \State
    Have $\rho$ join the order at time $T = \max \{ t_i \in \mathbf t \suchthat t_i \le d^*\}$.
    \label{line:T}

    \EndWhile \label{line:end}

    \EndFor

    \State   Let $\overline S$ denote the resulting schedule for $\overline \calI$.  
    Return $S = \big\{ (\lceil t \rceil, R)\suchthat (t,R)\in \overline S\big\}$.

\end{algorithmic}
\hrulefill
\hrulefill
\caption{ $\round_\distribution(\COST, \costVector, \demands, \mbfx)$
  randomly rounds the continuous-time fractional solution $\overline\mbfx$.
}\label{fig:round}
\end{minipage}
\end{figure}


\begin{proof}(of \lref[Lemma]{lemma: round})
  {\em Correctness and feasibility.}
  We claim first that the number $\GOTstop$ of order times 
  in line~\ref{line:t} has finite expectation; in other words,
  	with probability 1, $\GOTstop$ is finite --- line~\ref{line:t} finishes.
This follows from standard calculation:
    $\GOTstop$
    is the number of samples taken before the sum of
    the samples exceeds $\MEAS{\maxDeadline}-1$.
    Since $\E[\distribution] \ge \statistic p > 0$,
    there exists $\epsilon>0$ such that $\Pr[s \ge \epsilon] \ge \epsilon$.
    Thus, the expected number of samples needed 
    to increase the sum by $\epsilon$ is at most $1/\epsilon$.
    By linearity of expectation,
    the expected number of samples needed to increase 
    the sum by $\MEAS{\maxDeadline}$ is at most
    $\MEAS{\maxDeadline}/\epsilon^2$.
    Thus, $\E[\GOTstop] \le \MEAS{\maxDeadline}/\epsilon^2<\infty$.

  In the argument below (for cost estimation)
	we show  that each iteration of the inner loop on line~\ref{line:while}
  satisfies some not-yet-satisfied demand $(\rho,r,d)$.
  Thus, with probability 1, $\round_\distribution$ terminates.
  By inspection, it does not terminate until all demands are satisfied.
  Thus, with probability 1, $\round_\distribution$ 
  returns a feasible solution.


\emmedparagraph{Cost of the schedule.} 
We use the following basic properties of $\overline\mbfx$.
\begin{enumerate}
  \item Over any interval $(t,t'] \subseteq \overline\universe$,
    each retailer $\rho$'s take is at most the amount shipped.
    This follows directly from LP constraint~\eqref{eqn: density bound}.
  \item For each demand $(\rho, r, d)$,
    retailer $\rho$'s take over the demand period $(r-1,d]$ is at least 1.
	This holds because the take equals $\sum_{t=r}^d x_t$, 
    which is at least 1 by LP constraint~\eqref{eqn:  satisfy bound}.
    Consequently, the amount shipped over $(r-1,d]$ is also at least 1.
\end{enumerate}

  The given fractional solution $\mbfx$ 
  has warehouse cost $\COST\, \MEAS{\maxDeadline}$
  and retailer cost $\sum_{\rho} \costof{\rho}\, \AREA{\maxDeadline}$.
  (Recall that $\MEAS{\maxDeadline}$ is the amount $\overline\mbfx$ ships up to time $\maxDeadline$ 
  while $\AREA{\maxDeadline}$ is $\rho$'s take up to time $\maxDeadline$.)

  The algorithm's schedule $S$ has warehouse cost $\COST\, \GOTstop$
  and retailer cost $\sum_{\rho} \costof{\rho}\, \LOTstop$,
  where $\GOTstop$ is the number of orders placed in line~\ref{line:t}
  and $\LOTstop$ is the number of those orders joined by $\rho$
  in lines~\ref{line:join}--\ref{line:end}.

  We show $\E[\GOTstop] \le \MEAS{\maxDeadline}/\statistic p$
  and $\E[\LOTstop] \le \AREA{\maxDeadline}/\statistic p$.
  By linearity of expectation, these bounds imply that
  the schedule's expected cost is a most $1/\statistic p$ times the cost of $\mbfx$,
  proving the lemma.

  \smallskip

  First analyze $\E[\GOTstop]$, the expected number of samples until the sum
  of the samples exceeds \mbox{$\MEAS{\maxDeadline}-1$}. Clearly $\GOTstop$ is a
  stopping time.\footnote {That is, for any $i\in\nat$, the event $\GOTstop =
  i$ is determined by the first $i$ samples.} As noted previously, $\GOTstop$
  has finite expectation. So, by Wald's Lemma (see \lref[Appendix]{sec: walds
  lemma}), the expectation of the sum of the first $\GOTstop$ samples
  is not smaller than the expectation
  of $I$ times the expectation of each sample: $\E[\sum_{i=1}^\GOTstop \sample
  i] \geq \E[\GOTstop]\, \E[p]$. The sum is at most $\MEAS{\maxDeadline}$,
  because the sum of the first $\GOTstop-1$ samples is less than
  $\MEAS{\maxDeadline}-1$ and the last sample is at most 1. Thus,
  $\MEAS{\maxDeadline} \ge \E[\GOTstop]\, \E[p]$. Rearranging, $\E[\GOTstop]
  \le \MEAS{\maxDeadline}/\E[p] \le \MEAS{\maxDeadline}/\statistic p$, as
  desired.

  \smallskip

  Next analyze $\E[\LOTstop]$.
  Fix a retailer $\rho\in[m]$.
  Focus on the inner loop, lines~\ref{line:join}--\ref{line:end}.
  For each iteration $j\in[\LOTstop]$,  let $d^*_j$
  and $T_j$ denote, respectively, the value of $d^*$ and $T$
  in iteration $j$.
  The order that $\rho$ joins at time $T_j$
  indeed satisfies that iteration's unmet demand $(\rho, r, d^*_j)$,
  because $\rho$'s take over the period $(r-1,d^*_j]$ is at least 1,
  so the amount shipped over the period is at least 1,
  so, by the choice of $\mathbf t$ in line~\ref{line:t},
  the period has to contain some order time $t_i$,
  and $T_j \in [t_i, d^*_j]$.
  Then, by a standard induction on $j$, after $\rho$ joins the order at time $T_j$,
  all of $\rho$'s demands whose demand periods overlap $(0,T_j]$ are satisfied.
  Hence $\rho$'s order times are strictly increasing:
  $T_1 < T_2 < \cdots < T_{\LOTstop}$.

  Consider any non-final iteration $j$ of the loop.
  Define $\psi_j = (\sample 1, \sample 2, \ldots, \sample {k})$,
  the {\em state at the end of iteration $j$},
  to be the first $k$ samples drawn from $\distribution$,
  where $k=k_j$ is the number of samples
  needed to determine $T_j$ in iteration $j$.
  Explicitly, $k$ is determined by the condition

  \begin{equation}\label{eqn:kjt}
    t_{k-1}\,\le\, d^*_j\, <\, t_{k},
  \end{equation}

\noindent
  which implies $T_j = t_{k-1}$.
  (The sample $\sample {k}$ is included in $\state j$ because,
  for $T_j$ to be the maximum order time less than or equal to $d^*_j$,
  the order time {\em following} $T_j$ must exceed $d^*_j$.)
  When we look at the related warehouse shipments, 
  $t_{k-1} \leq d^*_j < t_k$ implies $\MEAS{t_{k-1}} \leq\, \MEAS{d^*_j} \, 
  \leq\, \MEAS{t_k}$. 
  The last inequality is in fact strict, 
  because the algorithm chooses $t_k$ minimally.
  Since $\overline\mbfx$ ships each $\sample i$ over the interval $(t_{i-1},t_i]$,
  the last relation implies that

  \begin{equation}\label{eqn:kj}
  \sum_{i=1}^{k-1} \sample i\,\le\, \MEAS{d^*_j}\, <\, \sum_{i=1}^{k} \sample i.
  \end{equation}
	
  Define $\rho$'s {\em take during iteration $j$},
  denoted $X_j$, to be $\rho$'s take over the interval $(T_{j-1},T_j]$
  (interpreting $T_0$ as 0).
  To finish the proof,  consider the sum $\sum_{j=1}^{\LOTstop-1} X_j$,
  that is, $\rho$'s take up to the start of the last iteration.

  The sum's upper index~$\LOTstop - 1$ is a stopping time.
  (Indeed, $\state j$ determines 
  which of $\rho$'s demands remain unsatisfied at the start of iteration $j+1$.
  Iteration $j+1$ will be the final iteration $\LOTstop$
  iff those unsatisfied demands can be satisfied by a single order.
  Thus, $\state j$ determines whether $\LOTstop - 1 = j$.)
  Clearly $\LOTstop-1$ has finite expectation.
  (Indeed, $\LOTstop$ is at most the number of demands.)

  We claim that expectation of each term $X_j$ in the sum,
  given the state at the start of iteration~$j$, is at least $\statistic p$:

  \begin{equation}\label{eqn:claim}
    \E[X_{j}~|~\state {j-1}]~\ge~ \statistic p.
  \end{equation}

  Before we prove Claim~\eqref{eqn:claim},
  observe that it implies the desired bound on $\LOTstop$, as follows.
  The upper index $\LOTstop-1$ of the sum $\sum_{j=1}^{\LOTstop-1} X_j$
  is a stopping time with finite expectation,
  and the conditional expectation of each term is at least $\statistic p$,
  so, by Wald's Lemma (see \lref[Appendix]{sec: walds lemma}),
  the expectation of the sum is at least $\E[\LOTstop-1] \statistic p$.
  On the other hand,  the value of the sum never exceeds $\AREA{\maxDeadline}-1$.
  (Indeed,   at the start of the last iteration $j=\LOTstop$, 
  some demand $(\rho, r, d)$ remains unsatisfied,
  and that demand, which has total take at least 1,
  does not overlap $(0,T_{j-1}]$,
  so $\rho$'s take up to time $T_{j-1}$ can be at most 1 less than the total take.)
  Thus, $\E[\LOTstop -1]\statistic p \le \AREA{\maxDeadline}-1$.
  Since $\statistic p \le 1$, this implies the desired bound $\E[\LOTstop] \le \AREA{\maxDeadline}/\statistic p$.

  To finish, we prove Claim~\eqref{eqn:claim}.
  Fix any state $\state {j-1}$ and  let $k=k_{j-1}=|\state {j-1}|$.
  Consider iteration $j$.
  Note that $\state {j-1}$ determines both $T_{j-1}$ and $d^*_j$.
  Call the samples in $\state{j}$ but not in $\state {j-1}$ {\em newly exposed}.
  Crucially, $\state {j-1}$ does not condition the newly exposed samples.
  Let random variable~$h = k_j-k_{j-1}$ be the number of newly exposed samples.
  By Condition~\eqref{eqn:kj}, $h$ is the index such that
  \(
  \sum_{i=1}^{k+h-1} \sample i
  \,\le\, \MEAS{d^*_j}
  \,<\, 
  \sum_{i=1}^{k+h} \sample i.
  \)

  Define $z = \MEAS{d^*_j} - \sum_{i=1}^{k} \sample i$. Then $z\ge 0$ with probability~$1$ and $z$
  is determined by $\state {j-1}$.
  Using $s'_1, s'_2, \ldots, s'_h$ to denote the newly exposed samples (i.e., $s'_\ell = s_{k+\ell}$),
  the condition on $h$ above is equivalent to

  \[
  s'_1 + s'_2 + \cdots + s'_{h-1} \,\le\, z \,<\,   s'_1 + s'_2 + \cdots + s'_{h-1} + s'_h.
  \]

\noindent
  That is, the iteration exposes new samples just until their sum exceeds $z$.
  Upon consideration, this process corresponds to a play of the tally game with threshold $z$,
  in the definition of the statistic $\statistic p$.
  (See \lref[Figure]{fig:tally}.)
  \begin{figure}[t]
    \centering
    \includegraphics[height=.93in]{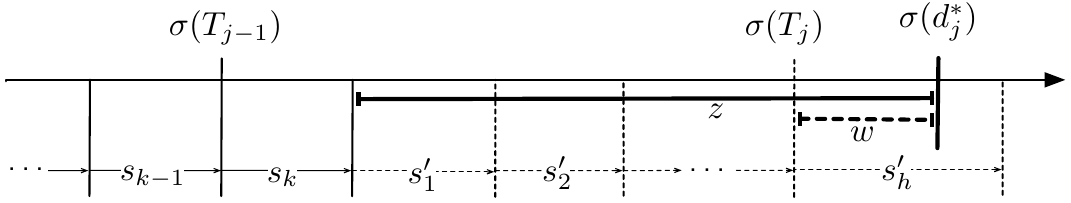}
    \caption{The increase in $\rho$'s take in each iteration corresponds to a play of the tally game.}
    \label{fig:tally}
  \end{figure}
 Recall from the definition of $\statistic p$
  that the {\em waste} is $w = z-(s'_1+s'_2+\cdots+s'_{h-1})$.

  By inspection, using that $\sum_{i=1}^{k+h-1} \sample i = \MEAS{T_j}$,
  the waste $w$ in this setting equals $\MEAS{d^*_j} - \MEAS{T_j}$,
  so $\rho$'s take $X_{j}$ during the iteration,  that is, $\AREA{T_j} - \AREA{T_{j-1}}$, equals

  \begin{align*}
  [\AREA{d^*_j} - \AREA{T_{j-1}}]
  -
  [\AREA{d^*_j} - \AREA{T_j}]
  ~&\ge~
  1
  -
  [\AREA{d^*_j} - \AREA{T_j}]
		\\
  ~&\ge~
  1 - [\MEAS{d^*_j} - \MEAS{T_j}]
  ~=~ 
  1-w.
  \end{align*}

  The first inequality holds because,
  as observed previously, $\rho$'s take over interval $(T_{j-1},d^*_j]$ is at least 1.
  The next inequality holds because
  $\rho$'s take over $(T_j, d^*_j]$ is at most the amount shipped.

  Recall that, by definition of $\statistic p$,
  the expectation of $(1-w)$ is at least $\statistic \distribution$.
  Thus, the inequality above implies Claim~\eqref{eqn:claim} --- that
  the conditional expectation of each $X_{j}$ is at least $\statistic p$.
\end{proof}

The next utility lemma
says that, in analyzing the expected waste in the tally game,
it is enough to consider thresholds $\threshold$ in $[0,1)$.

\begin{lemma}\label{lemma: 01}
  For any distribution $p$ on $[0,1]$,
  $\sup_{\threshold\ge 0} \EWaste \distribution \threshold  ~=~\sup_{\threshold\in[0,1)} \EWaste \distribution \threshold $.
\end{lemma}

\begin{proof}
  Play the tally game with any threshold $\threshold \geq 1$.  
  Consider the first prefix sum $\sample 1+\sample 2+\cdots+\sample h$
  of the samples, such that the ``slack'' 
  $\threshold - (\sample 1 + \sample 2 + \cdots + \sample h)$ is less than 1.
  Let random variable~$\threshold'$ be this slack.  Note that $\threshold'\in[0,1)$. 
  For any value $u\in[0,1)$,
  the expected waste conditioned on the event ``$\threshold'=u$'' is $\EWaste \distribution u$,
  which is at most $\sup_{y\in[0,1)} \EWaste \distribution y$.
  Thus, for \emph{any} threshold $z\ge 1$, $\EWaste \distribution \threshold $ 
  is at most $\sup_{y\in[0,1)} \EWaste \distribution {y}$. 
\end{proof}


\subsection{Upper bound of \texorpdfstring{$e/(e-1) \approx 1.582$}{e/(e-1) = 1.582}}

Consider the specific probability distribution $\distribution$ on $[0,1]$
with probability density function $\distribution(y) = 1/y$ for $y\in [1/e,1]$
and $\distribution(y) = 0$ elsewhere.

\begin{lemma}\label{lemma: upper bound 1}
  For this distribution $p$, $\statistic \distribution ~\ge~ (e-1)/e ~=~ 1-1/e$.
\end{lemma}

\begin{proof}
  By \lref[Lemma]{lemma: 01},
  $\statistic \distribution$ is the minimum of $\Exp[\distribution]$ 
  and $1 - \sup_{\threshold \in [0,1)} \EWaste \distribution \threshold $.

  By direct calculation,
  $\Exp[\distribution] = \int_{1/e}^1 y\,\distribution(y)\, dy ~=~ \int_{1/e}^1 1\, dy = 1-1/e$.
  Now
  consider playing the tally game with threshold $\threshold$.
  If $\threshold \in[0, 1/e]$, then (since the waste is at most $\threshold$)
  trivially $\EWaste \distribution \threshold  \le \threshold \le 1/e$.
  So, consider any $\threshold\in [1/e,1]$.  Let $\sample 1$ be just the first sample.
  The waste is $\threshold$ if $\sample 1 > \threshold$ and otherwise is at most $\threshold-\sample 1$.
	Thus,  the expected waste is 

	\begin{align*}
		\EWaste \distribution \threshold \;
		\leq \; & \Pr[\sample 1 > \threshold] \cdot\threshold
		 		~+~ \Pr[\sample 1 \le \threshold]\cdot \Exp[\threshold - \sample 1 ~|~ \sample 1 \le \threshold]  \\
		= \; & \threshold ~-~ \Pr[\sample 1 \le \threshold]\cdot \Exp[\sample 1 ~|~ \sample 1 \le \threshold] \\
		= \; & \threshold - \int_{1/e}^\threshold y\, p(y)\, dy
	  ~=~
	  \threshold - \int_{1/e}^\threshold \, dy
	  ~=~ \threshold - (\threshold-1/e)
	  ~=~ 1/e.
	\end{align*}

Since  both $\Exp[\distribution]$ and
$1-\sup_{\threshold \in [0,1)}\EWaste{\distribution}{\threshold}$ are at least $1-1/e$, the lemma follows.
\end{proof}

From \lref[Lemma]{lemma: upper bound 1}~and~\lref[Lemma]{lemma: round},
the approximation ratio of our
algorithm $\round_\distribution$, with distribution $\distribution$ defined above,
is at most $1/\statistic{\distribution} = \frac{e}{e-1} \approx 1.582$.


\subsection{Upper bound of \texorpdfstring{$1.574$}{1.574}}

On close inspection of the proof of \lref[Lemma]{lemma: upper bound 1}, 
it is not hard to see that the estimate for the waste in that proof is likely not tight. 
The reason is that the proof estimates the waste based on just the first sample, 
while, for the distribution being analyzed,
there is non-zero probability that {\em two} samples
are generated before reaching the threshold, further reducing the waste.
To improve the upper bound,
we adjust the probability distribution (and the analysis) accordingly.

Define a probability distribution $\distribution$ on $[0,1]$, having a point mass at $1$, as follows.
Fix $\theta = 0.36455$ (slightly less than $1/e$).
Over the half-open interval $[0,1)$, the probability density function is

\[
\distribution(y) ~=~
\begin{cases}
  0                  
  & \text{for } y \in [0,\theta) 
  \\
  1/y
  & \text{for } y \in [\theta,2\theta) 
  \\
  \frac{1-\ln((y-\theta)/\theta)}{y}
  & \text{for } y \in [2\theta, 1).
\end{cases}
\]

Define the probability of choosing $1$ to be $1-\int_0^1 \distribution(y) \,dy \approx 0.0821824$.
Note that $\distribution(y)\ge 0$ for $y\in [2\theta,1)$ since $\ln((1-\theta)/\theta)\approx 0.55567$,
so $p$ is indeed a probability distribution on $[0,1]$.


\begin{lemma}\label{lemma: upper bound 2}
  The statistic $\statistic \distribution$ for this $\distribution$ is at least $0.63533 > 1/1.574$.
\end{lemma}

\begin{proof}
Recall that $\statistic \distribution = \min \{\Exp[\distribution], 1-\sup_{\threshold\in[0,1)} \EWaste \distribution \threshold\}$.
By calculation, the probability measure~$\mu$ induced by $\distribution$ has
$\mu[1] \approx 0.0821824$ and

\begin{equation*}  
  \mu[0,v) ~=~
  \begin{cases}
    0                 
    & \text{for } v \in [0,\theta) 
    \\
    \ln( v/\theta)
    & \text{for } v \in [\theta,2\theta) 
    \\
    \ln (v/\theta)   ~-~
    		\int_{2\theta}^{v} \frac{\ln((y-\theta)/\theta)}{y} \,dy
    & \text{for } v \in [2\theta, 1).
  \end{cases}
\end{equation*}

\noindent
  The following calculation shows $\Exp[\distribution]> 0.63533$:

  \begin{eqnarray*}
    \Exp[\distribution]&=&\mu[1]+\int_{\theta}^{1}y\distribution(y)dy\\
    &=&\mu[1]+\int_{\theta}^{1}dy-\int_{2\theta}^{1}\ln{\left((y-\theta)/\theta\right)}dy\\
    &=&\mu[1]+(1-\theta)- \left((y-\theta)\ln{\left((y-\theta)/\theta\right)}-y \right) \vertbar_{2\theta}^1\\
    &=&\mu[1]+2-3\theta-  (1-\theta)\ln{\left((1-\theta)/\theta\right)}\\
    &>& 0.0821+2-3\cdot 0.36455-(1-0.36455)\cdot 0.5557\\
    &>&0.63533.
  \end{eqnarray*}

  To finish, we show $\sup_{\threshold\ge 0} \EWaste \distribution \threshold = \theta \le 1 - 0.63533$.
  
  By \lref[Lemma]{lemma: 01}, assume that $\threshold\in [0,1)$.
  In the tally game defining $\EWaste \distribution \threshold$,
  let $s_1$ be the first random sample drawn from $\distribution$.
  If $s_1 > \threshold$, then the waste equals $\threshold$.
  Otherwise, the process continues recursively with $\threshold$ replaced by $z' = \threshold-s_1$.
  This gives the recurrence

\begin{equation*}
\EWaste \distribution \threshold  ~=~ \threshold\,\mu[\threshold,1] 
							+ \int_0^\threshold \EWaste \distribution {\threshold-y} \distribution(y)\,dy
  						~=~ \threshold\,\mu[\threshold,1] 
							+ \int_\theta ^\threshold \EWaste \distribution {\threshold-y}\, \distribution(y)\,dy.
\end{equation*}

\noindent
Break the analysis into three cases, depending on the value of $z$.

\medskip
\noindent
\mycase{1} $\threshold\in [0,\theta)$. In this case $\mu[\threshold,1] = 1$, 
  so $\EWaste \distribution \threshold = \threshold \le \theta$.

\medskip
\noindent
\mycase{2} $\threshold\in [\theta,2\theta)$.
  For $y\in [\theta, \threshold]$, we have $\threshold-y\in[0,\theta]$,
  so, by Case~1, $\EWaste \distribution {\threshold-y} = \threshold-y$.
  Using the recurrence, 

\begin{align*}
	\EWaste \distribution \threshold ~&=~ 
   \threshold \mu[\threshold,1]+\int_{\theta}^\threshold (\threshold-y)\distribution(y)\,dy
\\
   ~&=~ \threshold\left(1-\int_{\theta}^\threshold \distribution(y)dy\right)+\int_{\theta}^\threshold (\threshold-y)\distribution(y)dy
   ~=~ \threshold -\threshold+\theta=\theta.
\end{align*}
 
\medskip
\noindent
\mycase{3} $\threshold\in [2\theta,1)$.
 For $y\in [\theta, \threshold]$, we have $\threshold-y\in [0,2\theta]$,
 so, by Cases~1~and~2 and the recurrence,

\begin{align*}
\EWaste \distribution \threshold
   &~=~ \threshold \mu[\threshold,1]
				+\int_{\theta}^{\threshold-\theta} \theta \distribution(y)dy
				+\int_{\threshold-\theta}^{\threshold} (\threshold-y)\distribution(y)dy
\\
  &~=~ \threshold-\threshold\int_{\theta}^{\threshold-\theta}\distribution(y)dy	
				+\int_{\theta}^{\threshold-\theta} \theta \distribution(y)dy
				-\int_{\threshold-\theta}^{\threshold} y\distribution(y)dy
\\
   &~=~\threshold-(\threshold-\theta)\ln{\frac{\threshold-\theta}{\theta}}
				-\int_{\threshold-\theta}^{\threshold} y\distribution(y)dy
\\
   &~=~ \threshold-(\threshold-\theta)\ln{\frac{\threshold-\theta}{\theta}}
				-\int_{\threshold-\theta}^{\threshold} dy
				+\int_{2\theta}^{\threshold}  \ln{\left((y-\theta)/\theta\right)} dy
\\
   &~=~ (\threshold-\theta)\left(1-\ln{\frac{\threshold-\theta}{\theta}}\right)
			+\int_{2\theta}^{\threshold}  \ln{\left((y-\theta)/\theta\right)} dy
\\
 &~=~(\threshold-\theta)\left(1-\ln{\frac{\threshold-\theta}{\theta}}\right)
				+\left(y-\theta)\cdot (\ln{\left((y-\theta)/\theta\right)}-1\right)\vertbar_{2\theta}^\threshold
\\
   &~=~(\threshold-\theta) 
        \left(
          1-\ln{\frac{\threshold-\theta}{\theta}}
          +\ln{\frac{\threshold-\theta}{\theta}}-1
        \right)
        + \theta 
        =\theta.
\end{align*}

Thus, in all cases, $\EWaste{\distribution}{\threshold}\le\theta$,
completing the proof.
\end{proof}


\begin{theorem}\label{thm: upper bounds}
  $\JRPD$ has a randomized polynomial-time $1.574$-approximation algorithm,
  and the integrality gap of the LP relaxation is at most $1.574$.
\end{theorem}

\begin{proof}
  By \lref[Lemma]{lemma: upper bound 2},
  for any fractional solution $\mbfx$, Algorithm~$\round_p$
  (using the probability distribution $\distribution$ from that lemma)
  returns a feasible schedule of expected cost at most $1.574$ 
  times $\costfn{\mbfx}$.  

  To see that the schedule can be computed in polynomial time,
  note first that the LP relaxation can be solved in polynomial time.
  The optimal solution $\mbfx$ is minimal,
  so each $x_t$ is at most 1, which implies that $\MEAS{\maxDeadline} = \sum_{t} x_t$ 
  is at most the number of demands, $n$.
  In Algorithm~$\round_p$,
  each sample from the distribution $\distribution$ from \lref[Lemma]{lemma: upper bound 2}
  can be drawn in polynomial time.
  Each sample is~$\Omega(1)$, and the sum of the samples 
  is at most $\MEAS{\maxDeadline} \le n$,
  so the number of samples is $O(n)$.
  In the inner loop of the algorithm
  (starting at line~\ref{line:join}),
  for each retailer~$\rho$, the subset of orders joined
  can be computed in time $O(n)$, by amortization, so the total time for
  this step is $O(mn)$, where $m$ is the number of retailers.
\end{proof}



\section{Upper Bound of 1.5 for Equal-Length Periods}
\label{sec: unit period upper bounds}


In this section, we present a 1.5-approximation algorithm for the case where all
the demand periods are of equal length.  In this section release times and deadlines 
are arbitrary rational numbers, and all demand periods have length $1$.

Denote the input instance by $\calI$. Let the \emph{width} of the instance be the
difference between the maximum deadline and the minimum release time.
The building block of our approach is an algorithm that creates an
optimal solution to an instance of width at most $3$.  Later, we divide~$\calI$
into overlapping sub-instances of width $3$, solve each of them optimally, and
show that aggregating their solutions gives a $1.5$-approximation.


\begin{lemma}
\label{lem:3-bounded-optimal}
A solution to any instance $\calJ$ of width at most $3$ consisting of unit-length
demand periods can be computed in polynomial time.
\end{lemma}

\begin{proof}
Shift all demands in time, so that $\calJ$ is entirely contained in
interval~$[0,3]$. Recall that $\COST$ is the warehouse ordering cost and $\costof{\rho}$ is the
ordering cost of retailer $\rho \in[m]$. Without loss of generality,
assume that $m\ge 1$ and each retailer has at least one demand.

Let $d_{\min}$ be the first deadline of a demand from $\calJ$ and $r_{\max}$
the last release time. If $r_{\max} \leq d_{\min}$, then placing one order at
any time from $[r_{\max},d_{\min}]$ is sufficient. Its cost is
then equal to $\COST + \sum_{\rho} \costof{\rho}$, which is clearly equal to
the optimum value in this case.

Now focus on the case $d_{\min} < r_{\max}$. 
Any feasible solution has to place an order at or before $d_{\min}$ and at or
after $r_{\max}$. Furthermore, by shifting these orders, assume that the
first and last orders occur exactly at times $d_{\min}$ and $r_{\max}$, respectively.

\begin{figure}[t]
\begin{center}
\includegraphics[width=5in]{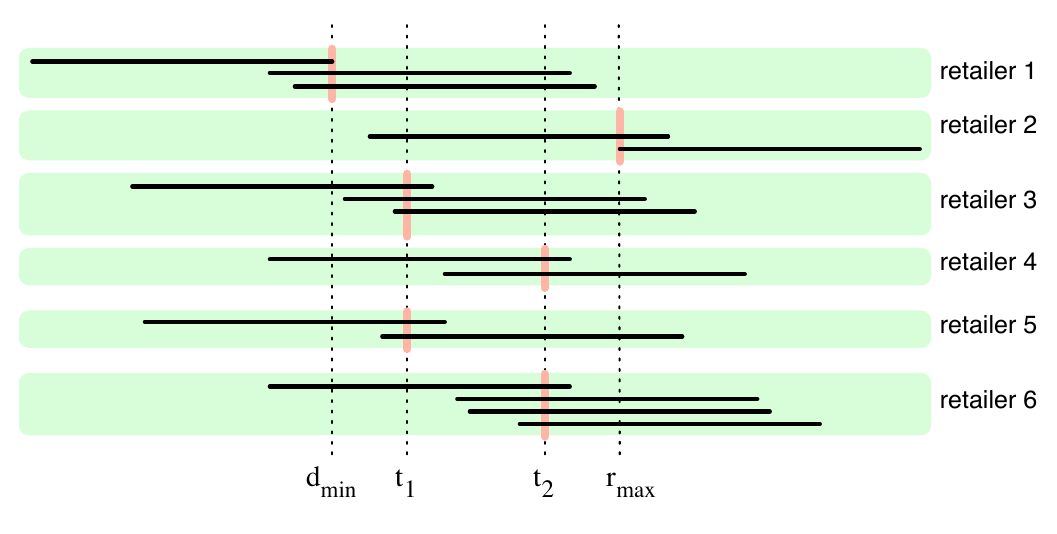}
\caption{An example of an instance and a schedule. Dashed vertical lines represent
warehouse orders, with thick segments indicating which retailers join these
orders. For example, retailer $1$ joins the order at
time $d_{\min}$, retailer $2$ joins the order at time $r_{\max}$,
and retailer $3$ joins the order at time $t_1$.} 
\label{fig: uniform example}
\end{center}
\end{figure}

The problem is thus to choose a set $T$ of warehouse ordering times that
contains $d_{\min}$, $r_{\max}$, and possibly other times from the interval
$(d_{\min},r_{\max})$, and then to decide, for each retailer $\rho$, 
which warehouse orders it joins. Note that $r_{\max}-d_{\min} \leq 1$, and
therefore each demand period contains $d_{\min}$, $r_{\max}$, or
both.  Hence, all demands of a retailer $\rho$ can be satisfied by joining the
warehouse orders at times $d_{\min}$ and $r_{\max}$ at additional cost of $2\costof{\rho}$. 
It is possible to reduce the retailer ordering cost to $\costof{\rho}$ if
(and only if) there is a warehouse order that occurs within $\demands_\rho$, where
$\demands_\rho$ is the intersection of all demand periods of retailer $\rho$.  (To
this end, $\demands_\rho$ has to be non-empty.) 

Hence, the optimal cost for $\calJ$ can be expressed as the sum of four parts: 
\begin{description}
	\item{(i)} the unavoidable ordering cost $\costof{\rho}$ for each retailer $\rho$, 
	\item{(ii)} the additional ordering cost $\costof{\rho}$
		for each retailer $\rho$ with empty $\demands_\rho$,
 	\item{(iii)} the total warehouse ordering cost $\COST \cdot |T|$, and 
	\item{(iv)} the additional ordering cost $\costof{\rho}$ for each retailer $\rho$ whose
		$\demands_\rho$ is non-empty and does not contain any ordering time from $T$.  
\end{description}
As the first two parts of the cost are independent of $T$, focus on minimizing
the sum of parts (iii) and (iv), which we call the {\em adjusted cost}.  
Let $\adjustedcost(t)$ be the minimum
possible adjusted cost associated with the interval $[d_{\min},t]$ under the assumption that there
is an order at time $t$. Formally, $\adjustedcost(t)$ is the minimum, over all choices of
sets $T\subseteq [d_{\min},t]$ that contain $d_{\min}$ and $t$, of
$\COST\cdot |T| + \sum_{\rho\in Q(T)}\costof{\rho}$, where
$Q(T)$ is the set of retailers $\rho$ for which
$\demands_\rho\neq\emptyset$ and $\demands_\rho \subseteq [0,t] - T$.
(Note that the second term consists of expenditures that actually occur outside the 
interval $[d_{\min},t]$.)

As there are no $\demands_\rho$'s strictly to the left of $d_{\min}$, 
$\adjustedcost(d_{\min}) = \COST$.
Furthermore, $\adjustedcost(t)$ for any $t \in (d_{\min},r_{\max}]$ can
be expressed recursively using the value of $\adjustedcost(u)$, where
$u \in [d_{\min},t)$ is the warehouse order time immediately preceding $t$
in the set $T$ that realizes $\adjustedcost(t)$:
\[
	\adjustedcost(t) = \COST + \min_{u \in [d_{\min},t)} \Big( 
		\adjustedcost(u) + \sum_{\rho : \emptyset \neq \demands_\rho \subset (u,t)} \costof{\rho}
	\Big).
\]

The second term inside the minimum represents the cost of retailers
whose sets $\demands_\rho$ do not contain an order.
The minimum in the definition of $\adjustedcost(t)$
is determined by a $u$ that is the deadline of some demand. 
Restricting attention to $t$'s and $u$'s that are deadlines of the demands,
compute the relevant values of function $\adjustedcost(\cdot)$ 
by dynamic programming in polynomial time.  
Finally, the total adjusted cost is~$\adjustedcost(r_{\max})$.
The actual orders can be recovered by a standard extension of the dynamic program.
\end{proof}

We now show how to construct an approximate solution for the original instance $\calI$
consisting of unit-length demand periods. For $i \in \nat$, let $\calI_i$ be the
sub-instance containing all demands entirely contained in $[i,i+3)$.  By
\lref[Lemma]{lem:3-bounded-optimal}, an optimal solution for~$\calI_i$, denoted
$A(\calI_i)$, can be computed in polynomial time.  Let $S_0$ be the solution
created by aggregating $A(\calI_0), A(\calI_2), A(\calI_4), \ldots$
and $S_1$ by aggregating $A(\calI_1), A(\calI_3), A(\calI_5),
\ldots$.  Among solutions $S_0$ and $S_1$, output the one with the smaller cost.


\begin{theorem}\label{thm: equal, 1.5 approximation}
The above algorithm produces a feasible schedule of cost at most $1.5$ times
optimum.
\end{theorem}

\begin{proof}
Each unit-length demand of instance $\calI$ is entirely contained in some
$\calI_{2k}$ for some $k \in \nat$. Hence, it is satisfied in $A(\calI_{2k})$,
and thus also in $S_0$, which yields the feasibility of $S_0$.  An analogous
argument shows the feasibility of $S_1$.

To estimate the approximation ratio, fix an optimal solution $\OPT$ for
instance~$\calI$ and let $opt_i$ be the cost of $\OPT$'s orders in the interval
$[i,i+1)$.  Note that $\OPT$'s orders in $[i,i+3)$ satisfy all
demands contained entirely in $[i,i+3)$. Since $A(\calI_i)$ is an optimal solution
for these demands, we have $\costfn{A(\calI_i)} \leq opt_i + opt_{i+1} +
opt_{i+2}$ and, by taking the sum, $\costfn{S_0} + \costfn{S_1} \leq
\sum_{i} \costfn{A(\calI_i}) \leq 3 \cdot \costfn{\OPT}$. 
Therefore, at least one of the
two solutions ($S_0$ or $S_1)$ has cost not larger than $1.5 \cdot \costfn{\OPT}$.
\end{proof}



\section{Lower Bounds of 1.207 and 1.245}
\label{sec: lower bound 1.207}


In this section we prove the following lower bound on the integrality gap of the LP relaxation from 
\lref[Section]{sec: introduction}:
\begin{theorem}\label{thm: integrality gap 1.207}
The integrality gap of the LP relaxation is
at least $\half (1 + \sqrt 2) \geq 1.207$.
\end{theorem}
We then sketch a computer-assisted proof of a stronger lower bound: $1.245$.

Fix an arbitrarily large integer $\maxDeadline$.
Define universe $\universe = \{i/\maxDeadline \suchthat i\in \nat\}\cap[0,\maxDeadline]$
to contain the non-negative integer multiples of $1/\maxDeadline$ in the interval $[0,\maxDeadline]$.
(The restriction to multiples of $1/\maxDeadline$ is a technicality;
throughout, for intuition, one can consider instead $\universe = [0,\maxDeadline]$.)
Consider an instance with warehouse-order cost $\COST=1$ and two retailers,
where retailer 1 has order cost $\costof{1} = 0$
and retailer 2 has order cost $\costof{2} = \sqrt{2}+\epsilon$,
where $\costof{2}$ is a multiple of $1/\maxDeadline$ and $0 \le \epsilon < 1/\maxDeadline$.
Retailer 1 has a demand for every time interval of length 1;
retailer 2 has a demand for every time interval of length $\costof{2}$:

\begin{equation*}
	\demands ~=~
		\{(1, t, t+1)  \suchthat  \,t, t+1\in \universe\} 
                ~\cup~ \{(2,t,t+\costof{2})  \suchthat  \,t,t+\costof{2}\in \universe\}.
\end{equation*}

Intuitively, in any solution,
retailer~$1$ must join at least one order in every subinterval of length~$1$,
so the warehouse-order cost is at least $1$ per time unit.
Retailer $2$ must join at least one order in any subinterval of length $\costof 2$,
so his order cost (not including the warehouse-order cost)
is at least~$\costof 2$ for every $\costof 2$ time units, i.e., 1 per time unit.
Thus, even if the two retailers could coordinate orders perfectly,
the total cost would be at least $2$ per time unit.

We show next that, because the orders cannot be coordinated perfectly,
the cost of any solution is at least about $1+\costof 2 \approx 1+\sqrt 2> 2$ per time unit.

Throughout this section, $o(1)$ denotes a term that tends to zero as $\maxDeadline$ tends to infinity.


\begin{lemma}\label{lemma:lb}
For the above instance, the optimal cost is at least $(1+\sqrt 2 - o(1)) \maxDeadline$.
\end{lemma}
\begin{proof}
  Fix any schedule for the instance. 
  Partition the time interval $[0,\maxDeadline]$
  into half-open subintervals,
  separated 
  by the times of orders that retailer 2 joins.
  Consider any such subinterval $(t,t']$.
  That is, 
  retailer 2 joins an order at time $t$ (or $t=0$),
  and, during the subinterval $(t,t']$
  retailer 2 joins an order at time $t'$ and no other time.
  We argue that the cost per time unit during $(t,t']$ 
  is at least $1+\sqrt 2 - o(1)$.

  First consider the case that the schedule has an order during $(t,t')$.
  The order at time $t'$ costs $1+\costof 2$;
  the additional order during $(t,t')$ costs at least 1.
  The interval length $t'-t$ is at most $\costof 2 + 1/\maxDeadline$ 
  (otherwise it would contain a demand of retailer 2,
  which, by the choice of $t$ and $t'$, would be unsatisfied).
  Thus, the cost per unit time is at least
  $(1+1+\costof{2})/(\costof{2}+1/\maxDeadline) = 1+\sqrt 2 - o(1)$.

  In the remaining case there is no order during $(t,t')$.
  The interval length $t'-t$ is at most $1+1/\maxDeadline$
  (otherwise the interval would contain an unsatisfied demand of retailer 1). 
  The order at time $t'$ costs $1+\costof 2$.
  Thus, the cost per time unit is at least
  $(1+\costof 2)/(1+1/\maxDeadline) = 1+\sqrt 2 - o(1)$.

  The last subinterval $(t,t']$ has to end at time $\maxDeadline-\costof{2}$ or later,
  so, in each subinterval of $[0,\maxDeadline-\costof 2] = [0,(1-o(1))\maxDeadline]$ 
  the algorithm pays at least $1+\sqrt{2} - o(1)$ per time unit.
\end{proof}

Next we observe that there is a fractional LP solution $\mbfx$ that costs 2 per time unit:
for each $t\in \universe$, 
let $x_t = x^1_t = 1/\maxDeadline$ and  $x^2_t = 1/(\costof 2 \maxDeadline)$.

Recall that $\universe$ contains the integer multiples of $1/\maxDeadline$ in $[0,\maxDeadline]$.
By calculation, the LP solution is feasible.
(For each demand of retailer 1,
the demand period intersects $\universe$ in $\maxDeadline+1$ times,
and so is satisfied.
Likewise, for each demand of retailer 2,
the demand period intersects $\universe$
in $c_2\maxDeadline+1$ times,
and so is satisfied.)

Since $|\universe| = \maxDeadline^2+1$,
the cost of fractional solution $\mbfx$ 
is $2(\maxDeadline^2+1)/\maxDeadline = (2+o(1))\maxDeadline$.
By \lref[Lemma]{lemma:lb},
any integer solution has cost $(1+\sqrt 2-o(1))\maxDeadline$.
Since the term $o(1)$ can be made arbitrarily small by choosing $\maxDeadline$ large,
the integrality gap of the LP is at least $(1+\sqrt 2)/2$,
proving \lref[Theorem]{thm: integrality gap 1.207}.


\myparagraph{Increasing the lower bound by a computer-based proof.} 
We now sketch how to increase the lower bound slightly to $1.245$.
We reduce the problem to that 
of maximizing the minimum mean cycle in a finite configuration graph,
which we solve with the help of linear programming.

Let the universe be $\universe = [\maxDeadline]$,
where $\maxDeadline$ is an arbitrarily large integer (tending to infinity).
Fix a~vector $\duration{} \in \nat_+^m$.
(Later we choose $m=5$ and $\duration{} = (6, 7, 8, 9, 11)$.)
Focus on instances where, for each retailer $\rho\in [m]$,
the retailer has uniform demands --- one for every subinterval of length $\duration{\rho}$.
That is, the demand set $\demands$ is
\(
\demands = \big\{ (\rho, t, t+\duration{\rho}-1)  \suchthat  \rho\in [m];\, \{t, t+\duration{\rho}-1\} \subseteq \universe\big\}.
\)
\smallskip

Define the ``uniform'' fractional solution $\mbfx$
by $x^\rho_t = 1/\duration{\rho}$ for all $\rho \in [m]$
and $x_t = \max_\rho 1/\duration{\rho}$ for all $t\in \universe$.
This solution is feasible for the LP and has cost 
$(\COST /\min_{\rho} \duration{\rho} + \sum_{\rho} \costof{\rho}/\duration{\rho})\cdot\maxDeadline$.

To bound the integer schedules we use a configuration graph.
Given any feasible schedule, 
for each order in the schedule,
define the \emph{configuration} at the order time, $t$,
to be a vector $\sigma \in \nat_{+}^m$ where
$\sigma_\rho$ is the time elapsed since $\rho$ last joined an order,
up to and including time $t$.
(If retailer~$\rho$ has not yet joined any order by time $t$, take $\sigma_\rho = t$.)
Feasibility of the schedule
ensures that each configuration $\sigma$ satisfies
$\sigma_\rho < \duration{\rho}$ for all $\rho\in[m]$,
because otherwise one of retailer $\rho$'s demands would not be met.
Thus, there are at most $\prod_\rho \duration {\rho}$ distinct configurations.
These are the nodes of the configuration graph.

The edges of the graph model possible transitions from one order to the next.
Let $\sigma$ denote the configuration at some order time $t$.
Let $\sigma'$ denote the next configuration, at the next order time $t'>t$.
Let $R$ be the set of retailers in the order at time $t'$.
Then $\sigma_\rho' = 0$ if $\rho\in R$
and $\sigma_\rho' = \sigma_\rho + t'-t$ otherwise.
To ensure feasibility, for all $\rho\in [m]$, it must be that
$\sigma_\rho + t'-t \le \duration{\rho}$ (even for $\rho\in R$).
Without loss of generality,
assume that $t'$ is maximal subject to this constraint
(otherwise, delay the second order without increasing the schedule cost).
That is, $t' = t+ \confDeadline(\sigma)$, where 
$\confDeadline(\sigma) = \min\{\duration{\rho} - \sigma_\rho \suchthat \rho \in [m]\}$.
For each $\sigma$ and $\sigma'$ that relate as described above,
the configuration graph has a directed edge from $\sigma$ to $\sigma'$.
The {\em cost} of the edge is the cost of the corresponding order,
$\costfn {\sigma, \sigma'} = \COST+\sum_{\rho\in R} \costof{\rho}$.
Let $G=(V,E)$ be the subgraph induced by nodes reachable from
the start configuration $(0,\ldots,0)$.
Explicitly construct the graph $G$,
labeling each edge $(\sigma,\sigma')$
with its elapsed time $\confDeadline(\sigma)$,
order set $R(\sigma,\sigma')$,
and cost function $\costfn {\sigma, \sigma'}$.

\begin{figure}[t]
\hrulefill
\begin{alignat}{6}
  {\textrm{Given $G=(V,E)$ and $\duration{}$, 
      choose $\costVector$, $\COST$, $\costfn{}$, and $\Phi$
      to maximize $\lambda$ subject to}}\span\span\span\span\span\span \notag
        \\     
  \COST &\;\ge && 0
        \notag
  \\
        \costof{\rho} &\;\ge && 0
                      & &\forall \rho \in \{1, \ldots, m\}
       \notag
        \\  \textstyle
  \COST /\min_{\rho} \duration{\rho} + \sum_{\rho} \costof{\rho} / \duration{\rho}
        & \;\le && 1
        \label{lb2:opt}
        \\
        %
        \costfn{\sigma,\sigma'} 
        &  \;= &\;\;& \textstyle
        \COST + \sum_{\rho\in R(\sigma,\sigma')} \costof{\rho}
                      & ~~~~ &\forall
                                                                                        (\sigma,\sigma') \in E
       \label{lb2:cost} 
        \\    
        \Phi_{\sigma} + \costfn{\sigma,\sigma'} - \Phi_{\sigma'} 
        & \;\ge &&  \confDeadline(\sigma) \, \lambda
                      & &\forall
                                                                                        (\sigma,\sigma') \in E.
       \label{lb2:pi}
\end{alignat}
\vspace*{-4.5ex}

\hrulefill
\caption{A linear program to choose the costs 
  to maximize the integrality gap $\lambda$, given
  the configuration graph $G$ and demand durations $\duration{\rho}$.}\label{fig:gap}
\end{figure}

In the limit (as $\maxDeadline \rightarrow \infty$),
every schedule will incur cost at least $\lambda$ per time unit
as long as, for every cycle $C$ in this graph,
the sum of the costs of the edges in $C$
is at least $\lambda$ times the sum of the times elapsed on the edges in $C$.
The integrality gap is then at least $\lambda$ 
divided by the cost (per time unit) of the uniform fractional solution defined above.
Note that $\lambda$ is essentially the minimum mean cycle cost in $G$.

Given any fixed $m$ and vector $\duration{}$ of period durations,
the configuration graph is determined.
We will choose the costs
(the warehouse-order cost $\COST$ and each retailer cost $\costof\rho$)
to maximize the resulting value of $\lambda$,
subject to the constraint that the cost of the uniform fractional solution is at most 1.
The linear program (LP) in \lref[Figure]{fig:gap} does this.
The LP is based on the standard LP dual for minimum mean cycle,
but the edge costs are not determined --- 
they are chosen subject to appropriate side constraints.
Constraint~\eqref{lb2:opt} of the  LP
is that the uniform fractional solution costs at most 1 per time unit.

We implemented this construction and,
for various manually selected duration vectors $\duration{}$ with small $m$,
we solved the linear program to find the maximum $\lambda$.
For efficiency, we used the following observations to prune the configuration graph.
We ordered $\duration{}$ so that $\duration{1} = \min_\rho \duration{\rho}$.
Without loss of generality we constrained $\costVector_1$ to be $0$
(otherwise replace $\COST$ by $\COST+\costVector_1$ and $\costVector_1$ by 0;
by inspection this gives an equivalent LP).
Then, since $\costVector_1=0$, without loss of generality, 
we assumed that retailer 1 is in every order $R(\sigma,\sigma')$.
We pruned the graph further using similar elementary heuristics.

The best ratio we found was for $\duration{} = (6, 7, 8, 9, 11)$.
The pruned graph $G$ had about two thousand vertices.
$\COST$ was about $2.49$, $\costVector_1$ was 0, 
every other $\costVector_\rho$ was about $1.245$,
and $\lambda$ was just above 1.245.



\section{Lower Bound of 1.2 for Equal-Length Periods}
\label{sec: lower bound 1.2 for equal}



In this section we show an integrality gap 
for the linear program for {\JRPD} for instances with equal-length demand periods.
The gap is for an instance with three retailers.
Numbering them for convenience starting from $0$,
their order costs are $\costof{0}=\costof{1}=\costof{2} = \onethird$. 
The warehouse-order cost is $\COST=1$.

In the demand set $\demands$, all intervals (demand periods) have length $2$.
Choose some large constant~$\maxDeadline$ that is a multiple of $3$. 
As illustrated in \lref[Figure]{fig: uniform gap instance},
for $\rho = 0,1,2$, retailer $\rho$'s demand periods are
\begin{equation*}
[3i+\rho,3i+2+\rho] \quad\textrm{and}\quad (3i-\threehalves+\rho,3i+\half+\rho), 
			\quad\textrm{for}\ i = 0,...,\maxDeadline/3.
\end{equation*}
\begin{figure}[t]
\begin{center}
\includegraphics[width=6in]{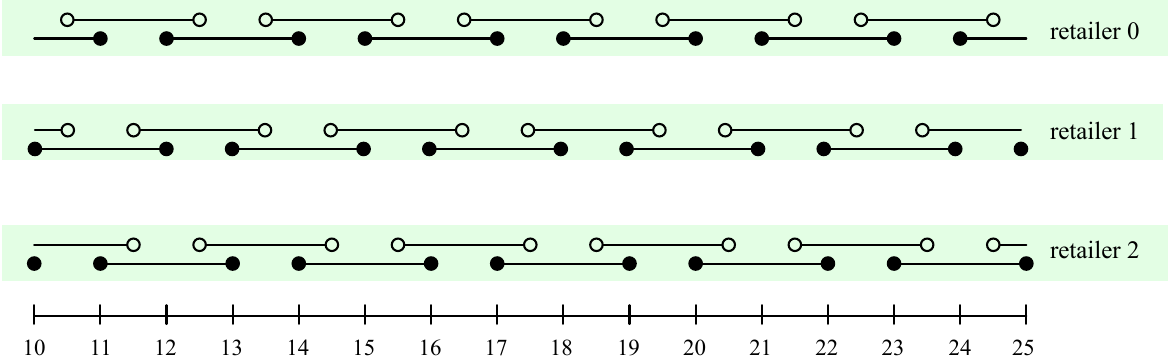}
\caption{The demand periods in $\demands$.}
\label{fig: uniform gap instance}
\end{center}
\end{figure}
To simplify the presentation, allow the demand periods to be either closed or
open intervals. (This is only for convenience: to ``close'' any open interval,
replace it by a closed interval slightly shifted to the right, with the shifts
increasing over time. Specifically, replace each interval $(a,a+2)$ by the
interval $[a+a/\maxDeadline^2,a+a/\maxDeadline^2+2]$. This preserves the
intersection pattern of all intervals; in particular any
two intervals $(a,a+2)$ and $(a+2,a+4)$ will remain disjoint after this
change. Therefore this change
does not affect the values of the optimal fractional and integral solutions
described below.)

This instance admits a fractional solution $\mbfx$ 
whose cost is $\fivesixths \maxDeadline + O(1)$:
For each integer time $t$,
place a $\half$-order that is joined by two retailers: 
the retailer $t \bmod 3$ whose closed interval starts at~$t$,
and the retailer $(t+1) \bmod 3$ whose closed interval ends at $t$
(let $s=(t-1)\bmod 3$; then $x_t^s = 0$, while $x_t = x^\rho_t = \half$ for $\rho\ne s$).
The cost of the $\half$-order at each time $t$ is $\half(2\cdot\onethird + 1) = \fivesixths$.

\smallskip

Now consider any integer solution $\hatmbfx$. 
Without loss of generality, assume that $\hatmbfx$ places orders only at integer times. 
(Any order placed at a fractional time $\tau$ can be shifted
either left or right to the first integer without changing the set of demands served.)

If a retailer $\rho$ has an order at time $t$, 
then its next order must be at $t+1$, $t+2$ or $t+3$, 
because for any $t$ the interval $(t,t+4)$ contains a demand period of retailer $\rho$. 
Thus, each retailer $\rho$ has to 
join some order in each triple $\{t+1,t+2,t+3\}$.
So, the retailer-cost per time unit for $\rho$ is at least $\costof{\rho}/3 = \frac{1}{9}$,
and the total retailer-cost per time unit is at least $\frac{3}{9} = \onethird$.

Similarly, if there is a warehouse order at time $t$, then the next order must be at
time $t+1$ or $t+2$, because the interval $(t,t+3)$ contains a demand period of some retailer.
So there must be some order in each pair $\{t+1,t+2\}$.
So the warehouse-order cost per time unit is at least $\COST/2 = \onehalf$.

In total, the total cost per time unit is at least $\onethird+\onehalf = \fivesixths$
(matching the cost of the fractional solution),
even if the retailer orders could be coordinated perfectly with the warehouse orders.
In the rest of this section,
we show that,
because perfect coordination is not possible,
the actual cost is higher.


Recall that (without loss of generality) in $\hatmbfx$ orders occur only at integer times.
For each $\rho$, call the endpoints of $\rho$'s closed intervals $\rho$'s \emph{endpoint times}
(these are times $t$ with $(t-\rho) \bmod 3 \in \{0,2\}$). 
Call the midpoints of $\rho$'s closed intervals $\rho$'s \emph{inner times}
(these are times $t$ with $(t-\rho) \bmod 3 = 1$).
Assume (without loss of generality by the feasibility and optimality of $\hatmbfx$)
that $\hatmbfx$ satisfies the following conditions:
\begin{enumerate}[label=(c\arabic*),ref=(c\arabic*),leftmargin=*]
	%
	\item\label{item:c1} For any $\rho$ and any pair of consecutive endpoint times $\{t,t+1\}$ of $\rho$, 
			$\rho$ joins an order at time $t$ or $t+1$ (because $\rho$
			has an open interval containing only integers $\{t,t+1\}$).
	\item\label{item:c2} If $t$ is an inner time of retailer $\rho$ and $\rho$ joins an order at
			$t$, then
			\begin{description}
				\item{(c2.1)} there is no order at time $t-1$ or $t+1$, and
				\item{(c2.2)} all retailers have orders at time $t$.
			\end{description}
			(For (c2.1): if there is an order at time $t-1$ or $t+1$,
			then retailer $\rho$ can be moved to that order from the order at time $t$.
			For (c2.2): 
			for each retailer $\rho'\ne \rho$ both time $t$ and either $t-1$ or $t+1$
                        are endpoint times,
                        but per (c2.1) there is no order at $t-1$ or at $t+1$,
                        so, by (c1), $\rho'$ must have an order at $t$.)
\end{enumerate}

\begin{figure}[t]
\begin{center}
\includegraphics[width=2.25in]{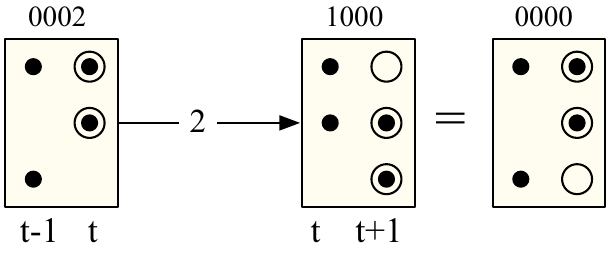}
\caption{Graphical representation of a transition.
Configuration $0002$ is on the left.   At time $t+1$ all retailers issue an order.
After removing spurious orders (keeping track of only the last one),
the new configuration is $1000$, which is equivalent (by symmetry) to $0000$.
The transition costs $1+\onethird\cdot 3 = 2$.
}
\label{fig: uniform gap configuration}
\end{center}
\end{figure}

The idea of the analysis is similar to the argument in
\lref[Section]{sec: lower bound 1.207} for the lower bound of $1.245$ for general instances:
represent the possible schedules by walks in a finite configuration graph.

Fix any feasible schedule.
At any integer time $t$, the {\em configuration} of the schedule at time $t$ 
is the 4-digit string $s\sigma_0\sigma_1\sigma_2$,
where $s = t\bmod 3$ and, for each retailer $\rho = 0,1,2$,
the elapsed time since the retailer last joined an order is $\sigma_\rho$.
Since the schedule is feasible, each $\sigma_\rho$ is in $\{0,1,2\}$,
so there are at most $3^4$ possible configurations.

Suppose a schedule is in configuration $s\sigma_0\sigma_1\sigma_2$ at time $t$,
then transitions to $s'\delta_0\delta_1\delta_2$ at time $t+1$.
Necessarily $s' = (s+1)\bmod 3$.
Let $R$ be the set of retailers (possibly empty) that join the order (if any) at time $t+1$. 
For each retailer $\rho$, (i) if $\rho\notin R$ then $\delta_\rho = \sigma_\rho+1$,
while (ii) if $\rho\in R$ then $\delta_\rho = 0$. 
Say a pair $s\sigma_0\sigma_1\sigma_2 \rightarrow s'\delta_0\delta_1\delta_2$
is a {\em possible transition} if the pair relates this way for some $R$.
The {\em cost} of the transition equals the cost of the order:
$0$ if $R=\emptyset$, or $1 + \onethird |R|$ otherwise.
(Here, unlike in \lref[Section]{sec: lower bound 1.207}, 
the elapsed time per transition is always 1, and $R$ can be empty.)

Represent each possible configuration graphically by a rectangle 
with a row for each retailer $\rho=0,1,2$.
Each row has two cells, representing times $t-1$ and $t$, respectively:
a circle in the first cell means $\sigma_\rho = 1$,
a circle in the second cell means $\sigma_\rho = 0$,
no circle means $\sigma_\rho = 2$.
A dot in the cell means that time is an endpoint time for the retailer;
no dot means the time is an inner time.
The dot pattern of any one row determines, and is determined by, $s$.
\lref[Figure]{fig: uniform gap configuration} shows an example of a single transition.

Any two configurations are {\em equivalent} 
if one can be obtained from the other by permuting 
the rows of the graphical representation (i.e., the retailers).
Each graphical representation has one row with a dot in both columns,
one row with a dot in the second column only,
and one row with a dot in the first column only.
Define the {\em canonical representative} of an equivalence class
to be the configuration in which these three rows are, respectively,
first, second, and third.
In all such configurations, $s$ is zero, so there are at most $3^3$ equivalence classes.

Now restrict the configurations further to those that are realizable in $\hatmbfx$,
in that they don't violate conditions \ref{item:c1}--\ref{item:c2}:
by \lref[Condition]{item:c1}, if a row has two dots, then one of the dots must be circled;
by \lref[Condition]{item:c2}, if a column has a dot-less circle, then all cells in the column have circles.
Note that the equivalence relation respects these conditions:
in a given equivalence class, either all configurations meet both conditions, or none do.

Finally, define graph $G$ to have a node for each realizable equivalence class.
For each possible transition $\sigma \rightarrow \sigma'$
between remaining configurations $\sigma$ and $\sigma'$,
add a directed edge in $G$ from the equivalence class of $\sigma$ to that of $\sigma'$.
Give the edge cost equal to the cost of the transition.
By a~routine but tedious calculation,
$G$ is the 10-node graph shown in \lref[Figure]{fig: uniform gap nfa}.
Each node is represented by its canonical representative.

\begin{figure}[t]
\begin{center}
\includegraphics[width=4.5in]{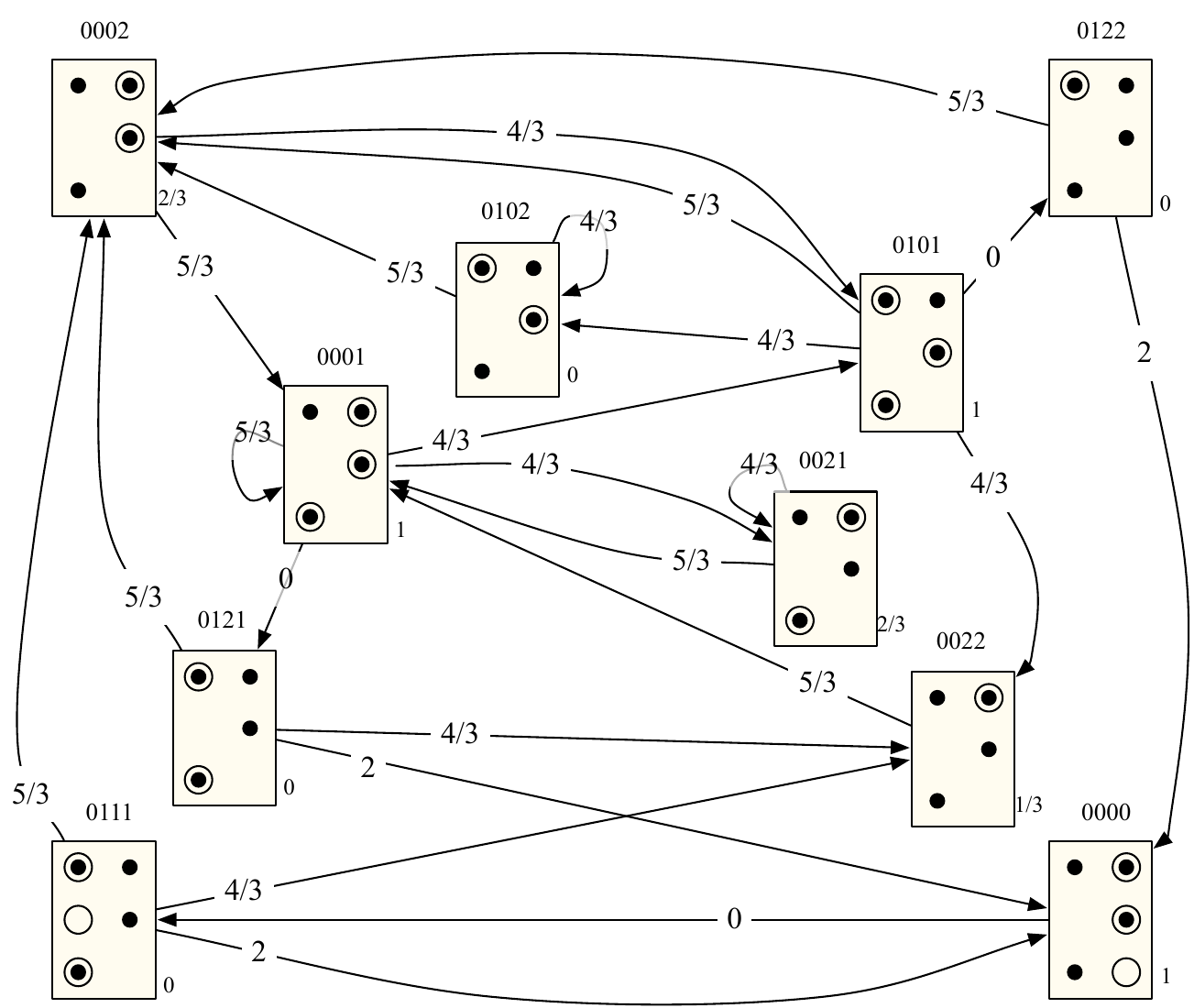}
\caption{The complete transition diagram.}
\label{fig: uniform gap nfa}
\end{center}
\end{figure}

Next we argue that every cycle in this graph has average edge cost at least 1.
Define the following potential function $\Phi(\sigma)$ on configurations:

\begin{center}\small
	\renewcommand\arraystretch{2}
\begin{tabular}{|c|c|c|c|c|c|c|c|c|c|c|} \hline
$\sigma$ 
	& ${0000}$ & ${0001}$ & ${0002}$ & ${0021}$ & ${0022}$ 
		& ${0101}$ & ${0102}$ & ${0111}$ & ${0121}$ & ${0122}$ \\ \hline
$\Phi(\sigma)$
	&  $1$ & $1$ & $\twothirds$ & $\twothirds$ & $\onethird$ 
			& $1$ & $0$ & $0$ & $0$ & $0$ \\ \hline
\end{tabular}
\end{center}

It is routine (if tedious) to verify that for each edge $\sigma \rightarrow
\sigma'$, its cost $\costfn{\sigma,\sigma'}$ satisfies
\begin{equation}
	\costfn{\sigma,\sigma'}  \ge \Phi(\sigma') - \Phi(\sigma) + 1.
		\label{eqn: uniform gap potential}
\end{equation}
For any path of length $\maxDeadline$ in $G$
(summing inequality~\eqref{eqn: uniform gap potential} along all edges on this path)
the cost of the path is $\maxDeadline - O(1)$.

The equivalence classes of the configurations of the schedule $\hatmbfx$
(one for each time $t\in[\maxDeadline]$)
form a path of length $\maxDeadline$ in $G$.
The cost of $\hatmbfx$ equals the cost of the path,
which must be at least $\maxDeadline - O(1)$.
Recalling that there is a fractional solution of cost $\fivesixths \maxDeadline+O(1)$,
this shows that the integrality gap is at least $\frac{6}{5} = 1.2$:

\begin{theorem}\label{thm: integrality gap 1.2 equal}
For instances with all demand periods equal,
the integrality gap of the LP at least~$1.2$.
\end{theorem}

(The bound in the above proof is tight:
the following cycles have average edge cost 1:
${0000}\to {0111}\to {0000}$,
~${0101}\to {0122}\to {0002}\to {0101}$, and
${0121}\to {0022}\to {0001}\to {0121}$.)



\section{APX-Hardness for Equal-Length Demand Periods}
\label{sec: apx hardness}



Let $\JRPDEfour$ be the restriction of $\JRPD$ where each retailer has at most
four demands and all demand periods are of the same length.  
In this section, we show that $\JRPDEfour$ is $\APX$-hard
by giving a PTAS-reduction from Vertex Cover in cubic graphs,
that is graphs with every vertex of degree three.
Vertex Cover is known to be $\APX$-complete for such graphs \cite{alimonti-apx-complete-cubic}.

Roughly speaking, given any cubic graph $G = (V,E)$ with $n$ vertices
and $m$ ($= 3n/2$) edges, the reduction produces an instance $\JG$ of $\JRPDEfour$, 
such that $G$ has a vertex cover of size $K$ 
iff $\JG$ has a schedule of cost $10.5 n + K + 6$.
Since any vertex cover has size at least $m/3 = n/2$,
this is a PTAS-reduction. 
Such reduction and a PTAS for $\JG$ would give a PTAS for Vertex Cover in degree-three graphs.


\myparagraph{Construction of instance $\JG$. }
Fix a given undirected cubic graph $G$ with vertex set $V=\{0,\ldots,{n-1}\}$
and edges $e_0,\ldots,e_{m-1}$.  
$\JG$ consists of $1 + m + n$ gadgets: 
one \emph{support gadget} $\SG$,
an \emph{edge gadget} $\EG_j$ for each edge~$e_j$, 
and a \emph{vertex gadget} $\VG_i$ for each vertex~$i$.
All retailer and order costs equal 1
($\COST=\costof{\rho} =1$)
and all demand periods have length $4m$.
All release times and deadlines are integers in the interval $[-4m,12m+1]$;
without loss of generality, restrict attention to schedules
with integer order times.
The gadgets are as follows.

\begin{figure}[t]
\begin{center}
\includegraphics[width=6in]{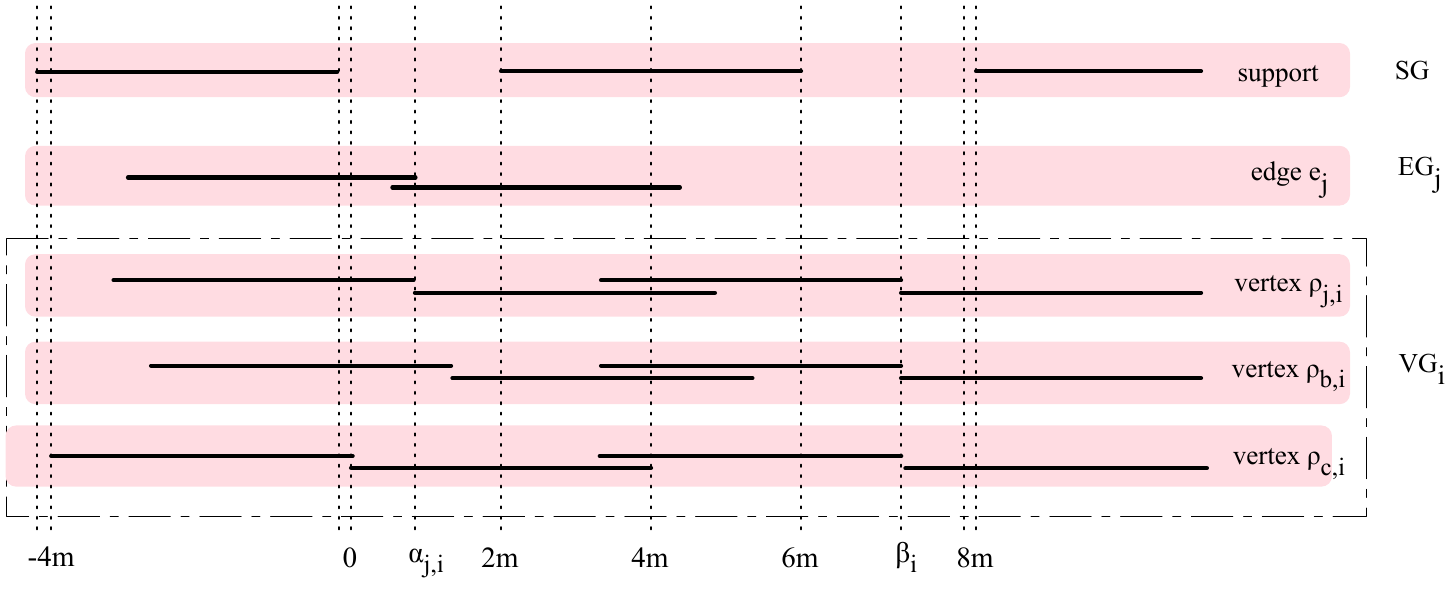}
\caption{The construction of instance $\JG$. The figure shows
the support gadget $\SG$, an edge gadget $\EG_j$ and a vertex gadget $\VG_i$
for a vertex $i$ with edges $e_j$, $e_b$, $e_c$.
Shaded regions represent retailers. Horizontal line segments represent 
demand periods.}
\label{fig: apx-hardness}
\end{center}
\end{figure}

\emmedparagraph{The support gadget $\SG$.}
$\SG$ has its own retailer,
the {\em support-gadget retailer}, having three demands
with periods $[-4m-1,-1]$, $[2m,6m]$, and $[8m+1,12m+1]$.
These periods are separated by two {\em gaps} of length $2m$.
Call the times $\{-1, 4m, 8m+1\}$ {\em support gadget} times;
orders at these times suffice to satisfy the three demands.

\emmedparagraph{Edge gadgets $\EG_j$.}
Each edge $e_j$ in $G$ has its own {\em edge} retailer, having two demands 
with, respectively, periods $[2j+1-4m,2j+1]$ and $[2j,2j+4m]$.
These demands can be satisfied with a~single order at time $2j$ or $2j+1$.
Think of these two times as being associated with this edge $e_j$, each associated
with one endpoint of $e_j$ (as explained below).
All such times are in the first gap, $[0,2m-1]$.

Let $e_j = \{i, i'\}$, that is $i,i'$ are the endpoints of edge $e_j$.
Intuitively, to satisfy $e_j$'s retailer cheaply, there can be an order at time $2j$ or at $2j+1$;
this models that $e_j$ can be covered by either of its two endpoints.
We associate one of the two times $2j$ or $2j+1$ 
(it does not matter which one) with $i$ and call it $\alpha_{j,i}$, 
while the other one is associated with $i'$ and, naturally, called $\alpha_{j,i'}$.

\emmedparagraph{Vertex gadgets $\VG_i$.}
For each vertex $i$, define its ``vertex'' time $\beta_i = 8m-i$.
(All such times are in the second gap, $[6m+1,8m]$.)
Do the following for each of vertex $i$'s edges.
Let $e_j$ denote the edge.
Add a new {\em vertex retailer} $\rho_{j,i}$,
having four demands with respective periods

\[
[\alpha_{j,i}-4m,\,\alpha_{j,i}],
~
[\alpha_{j,i},\,\alpha_{j,i}+4m], 
~
[\beta_i-4m,\,\beta_i],
\textrm{ and } [\beta_i,\,\beta_i+4m].
\]

Denote the periods of these four demands as
$\llQ_{j,i}$, $\lQ_{j,i}$, $\rQ_{j,i}$ and $\rrQ_{j,i}$, in the above order.
Note that $\llQ_{j,i}\cap \lQ_{j,i} = \braced{\alpha_{j,i}}$,
$\lQ_{j,i}\cap \rQ_{j,i} = [\beta_i-4m,\alpha_{j,i}+4m]\neq\emptyset$,
$\rQ_{j,i}\cap \rrQ_{j,i} = \braced{\beta_i}$, 
but otherwise the four demand periods are pairwise-disjoint.

The important property is that retailer $\rho_{j,i}$'s four demands can be satisfied
with two orders iff the two orders are at times $\alpha_{j,i}$ and $\beta_i$.
Also, if orders do happen to be placed at these two times, then
(because $\alpha_{j,i}$ is one of the two times belonging to $e_j$)
the order at time $\alpha_{j,i}$ can satisfy both demands of edge $e_j$'s retailer
with no additional warehouse cost for that retailer.

\lref[Figure]{fig: apx-hardness} illustrates the reduction.



\begin{lemma}
\label{lem:cover-to-schedule}
If $G$ has a vertex cover $U$ of size $K$, 
then $\JG$ has a schedule of cost at most $10.5n+K+6$. 
\end{lemma}

\begin{proof}
Let $U$ be a vertex cover of size $K$.

To construct the schedule for $\JG$, 
start with orders at the support-gadget times $\{-1, 4m, 8m+1\}$, 
each of which is joined by the support-gadget retailer.  This costs $6$.

Next, consider each vertex $i$.
If $i\notin U$, then have $i$'s vertex retailers 
(that is, retailers $\rho_{j,i}$ for $e_j\ni i$)
join the support-gadget orders at times $\{-1, 4m, 8m+1\}$.
This option increases the schedule cost by $3 \cdot 3 = 9$.

Otherwise (for $i \in U$),  create an order at time $\beta_i$. For each of
three $i$'s retailers $\rho_{j,i}$ create an~order at time $\alpha_{j,i}$, and
have the retailer join that order and the one at time $\beta_i$. The order at
time $\beta_i$ is shared between $i$'s three retailers, so that order costs $1
+ 3 = 4$. Each of the three orders created at times $\alpha_{j,i}$ costs $2$.
The total cost for the four orders is $3 \cdot 2 + 4 = 10$.

Next, consider each edge $e_j$.
As $U$ is a vertex cover, some vertex $i\in U$ covers $e_j$, that is $i\in U\cap e_j$.
By the construction of the $i$'s gadget, since $i\in U$,
there is already an order at $e_j$'s time $\alpha_{j,i}$.
Have edge $e_j$'s retailer join this order. Both demands of this retailer
will be satisfied, since they both contain $\alpha_{j,i}$.
The cost increases by 1 per edge.

Adding up the above costs,
the total cost is $6 + 9(n-K) + 10K + m = 10.5n+K+6$.
\end{proof}


\myparagraph{Recovering a vertex cover from an order schedule.}
We now show the converse: given any order schedule of cost $10.5n+K+6$ for
$\JG$, we can compute a vertex cover of $G$ of size $K$.
Recall that $\alpha_{j,i}$ denotes the time (either $2j$ or $2j+1$) that edge $e_j$ shares with endpoint $i$.

Say that an order schedule $S$ meeting the following desirable conditions is in {\em normal form}:
\begin{enumerate}[label=(nf\arabic*),ref=(nf\arabic*),leftmargin=*]
\item \label{item:nf1}
  In $S$, the support-gadget retailer joins orders at the support-gadget times $\{-1, 4m, 8m+1\}$.
\item \label{item:nf2}
  In $S$, for each edge $e_j=\{i,i'\}$, the edge's retailer joins an order at time $\alpha_{j,i}$ or $\alpha_{j,i'}$.
\item \label{item:nf3}
	For each vertex $i$, exactly one of the following two conditions holds:
		\begin{enumerate}[label=(\alph*),ref=\ref{item:nf3}\,(\alph*),leftmargin=*]
			\item\label{item:nf3.1}  each of $i$'s retailers $\rho_{j,i}$ 
                          joins orders at times $\beta_i$ and $\alpha_{j,i}$;
			\item\label{item:nf3.2} each of $i$'s retailers $\rho_{j,i}$
                          joins the support-gadget orders at times $\{-1, 4m, 8m+1\}$,
                          and $S$ has no order at time $\beta_i$ nor at any time $\alpha_{j,i}$.
		\end{enumerate}
\item\label{item:nf4}
  For each edge $e_j = \braced{i,i'}$, at least one of its endpoints $i,i'$ satisfies \lref[Condition]{item:nf3.1}.
\item \label{item:nf5} $S$ has no orders other than the ones described above.
\end{enumerate}

Given any feasible order schedule,
we can put it in normal form without increasing the cost:


\begin{lemma}
\label{lem:normalization}
Given any order schedule $S$ for $\JG$, one can compute in polynomial time
a normal-form schedule $S'$ whose cost is at most the cost of $S$.
\end{lemma}

\begin{proof}
Modify $S$ to satisfy \lref[Conditions]{item:nf1} through\lref{item:nf5} in turn,
maintaining feasibility without increasing the cost, as follows.

\noindent{\em\lref{item:nf1}}
Combine all orders in times $(-\infty,-1]$ into a single order at time $-1$.
By inspection, the earliest deadline of any demand is
the deadline of the first support-gadget demand, which is $-1$.
So this modification is {\em safe} --- it maintains feasibility without increasing the cost
--- and the support-gadget retailer joins the order at time $-1$.

Likewise, combine all orders in times $[8m+1,\infty)$ into a single order at time $8m+1$.
The last release time of any demand is
the release time of the last support-gadget demand, which is $8m+1$.
So this modification is also safe and the support-gadget retailer joins the order at time $8m+1$.

Finally, combine all orders in times $[2m,6m]$ into a single order at time $4m$.
The support-gadget has demand period $[2m,6m]$, so the support-gadget retailer
must join at least one order at some time in $[2m,6m]$. Thus
this modification does not increase the cost.
There are no deadlines in $[2m,4m)$ and no release times in $(4m, 6m]$, 
so the modification maintains feasibility.

The resulting schedule satisfies \lref[Condition]{item:nf1}.

\smallskip

\noindent {\em \lref{item:nf2}}
Consider any edge $e_j=\{i,i'\}$.
If the edge's retailer does not join an order at one of the times $\alpha_{j,i}$ or $\alpha_{j,i'}$
associated with $e_j$ then, by inspection of his demands, the retailer must join at least two orders.
Remove him from these two orders, reducing the cost by two,
and have him join a (possibly) new order at, time, say $\alpha_{j,i}$,
satisfying both his demands and increasing the cost by two or less.
The resulting schedule satisfies \lref[Conditions]{item:nf1} and\lref{item:nf2}.

\smallskip

\noindent {\em \lref{item:nf3}}
Consider any vertex $i$. Assume first that $S$ has an order at time $\beta_i$.
If any of vertex $i$'s retailers, say $\rho_{j,i}$, does not join the order at
time $\beta_i$, then move him from some order at any later time (there must be
one in his last demand period $\rrQ_{j,i} = [\beta_i,\beta_i+4m]$) to the
existing order at time $\beta_i$. Then, if retailer $\rho_{j,i}$ does not join
an order at time $\alpha_{j,i}$ (or there is no such order), he must
participate in at least two orders at times other than $\beta_i$; remove him
from these two orders and have him join a (possibly new) order at time
$\alpha_{j,i}$. Finally, remove the retailer from all orders other than those
at times $\beta_i$ and $\alpha_{j,i}$. These operations are safe, and now
vertex $i$ meets \lref[Condition]{item:nf3}.

In the other case $S$ has no order at time $\beta_i$.
Then each of vertex $i$'s retailers must join at least three orders. 
Remove each such retailer from all those orders
and add him instead to the existing orders at the support-gadget times $\{-1,4m,8m+1\}$.
Note that support-gadget times are different than all times $\alpha_{j,i}$.
This is safe, does not increase the cost,
and now vertex $i$ meets \lref[Condition]{item:nf3}.

As these operations affect only vertex retailers,
Conditions\lref{item:nf1} and\lref{item:nf2} still hold too.

\smallskip

\noindent{\em \lref{item:nf4}}
Consider any edge $e_j = \braced{i,i'}$. By \lref[Condition]{item:nf2}, there is an order at time
$\alpha_{j,i}$ or $\alpha_{j,i'}$. By symmetry, we can assume that there is an order at time
$\alpha_{j,i}$. If vertex $i$ satisfies \lref[Condition]{item:nf3.1}, we are done.
Otherwise, $i$ satisfies \lref[Condition]{item:nf3.2}. Remove each of $i$'s three retailers
from orders at times $-1,4m,8m+1$, reducing $i$'s cost by $9$. Then:
create an order at time $\beta_i$ and have them join this order (at cost $4$),
have retailer $\rho_{j,i}$ join the order at time $\alpha_{j,i}$ (at cost $1$),
and have the other two retailers $\rho_{j',i}$ and $\rho_{j",i}$
join (possibly new) orders at times $\alpha_{j',i}$ and $\alpha_{j",i}$ (at cost at most~$2$~each).
The total cost of these new orders is at most $9$, thus this modification does not
increase the overall cost, and afterwards $e_j$ satisfies \lref[Condition]{item:nf4}.

\smallskip

\noindent{\em \lref{item:nf5}}
Remove all retailers from orders not described above, then delete empty orders.
As the orders described above satisfy all the demands, this is safe.
Now \lref[Condition]{item:nf5} holds as well.
\end{proof}


\begin{lemma}
\label{lem:schedule-to-cover}
Given an order schedule $S$ for $\JG$
of cost $10.5n+K+6$, one can compute in polynomial time a vertex cover of $G$ of size $K$.
\end{lemma}

\begin{proof}
By \lref[Lemma]{lem:normalization}, without loss of generality we can assume
that that $S$ is in normal form.

By \lref[Condition]{item:nf1}, the cost for the support-gadget orders
(at times $\{-1,4m,8m+1\}$, but not yet counting the retailer cost for any vertices) is $3\cdot 2 = 6$.

By \lref[Condition]{item:nf3}, the cost associated with vertices is as follows.
Fix any vertex $i$. If $S$ has an order at time $\beta_i$ then
that order is joined by each of vertex $i$'s retailers, at cost $1+3 = 4$;
also, each of $i$'s retailers joins its own order (at a time $\alpha_{j,i}$ where $e_j\ni i$), which costs $2$.
Thus the cost associated with vertex $i$ is $4 + 2\cdot 3 = 10$.
Otherwise (that is, when $S$ has no order at time $\beta_i$),
each of $i$'s three retailers joins the three support-gadget orders,
so the cost associated with $i$ is $3\cdot 3 = 9$.
Putting it together, and letting $\ell$ be the number of vertices $i$ that have an order at time $\beta_i$,
the total cost associated with all vertices can be 
written as $10\ell + 9(n-\ell) = 9n + \ell$.

By \lref[Condition]{item:nf2}, the cost associated with edges is as follows.
For each edge $e_j=\{i,i'\}$, \lref[Condition]{item:nf4} guarantees that one of
$i$ or $i'$ satisfies \lref[Condition]{item:nf3.1}. By symmetry, assume that it is $i$.
Since $\rho_{j,i}$ already has an order at $\alpha_{j,i}$, we can
have retailer $e_j$ join this order at cost $1$. 
In total, the additional cost associated with the edge gadgets is $m = 1.5n$.

In sum, the schedule costs $6 + (9n+\ell) + 1.5n = 10.5n + \ell + 6$. Hence, $\ell = K$.

Now define $U$ to contain the $\ell$ vertices $i$ for which $S$ makes an
order at time $\beta_i$. For each edge $e_j=\{i,i'\}$,
by~\lref[Condition]{item:nf4}, there is an order at one of $e_j$'s associated times, say $\alpha_{j,i}$.
By~\lref[Condition]{item:nf3}, there is also an order at time $\beta_i$,
so, by definition, $U$ contains vertex $i$.
Thus, $U$ is a vertex cover.
\end{proof}

Here is the proof of $\APX$-hardness.
Recall that $\JRPDEfour$ is $\JRPD$
restricted to instances with equal-length demand periods and at most four demands.


\begin{theorem}\label{thm: apx-hardness}
$\JRPDEfour$ is $\APX$-hard.
\end{theorem}

\begin{proof}
Vertex Cover in cubic graphs is $\APX$-hard~\cite[\S 3]{alimonti-apx-complete-cubic}.
We give a PTAS-reduction from that problem to $\JRPDEfour$. 

Given any cubic graph $G$ with $n\ge 6$ vertices, and any $\epsilon>0$,
compute the instance $\JG$ from \lref[Lemma]{lem:cover-to-schedule}.
By inspection of the proof, in $\JG$ all demand periods have equal length
and (because $G$ has degree three) each retailer has at most four demands,
so $\JG$ is an instance of $\JRPDEfour$.

Now suppose we are given any $(1+\epsilon/24)$-approximate solution $S$ to $\JG$.
From $S$, compute a~vertex cover $U$ for $G$ using the computation from \lref[Lemma]{lem:schedule-to-cover}.  
The computations of $\JG$ from $G$, and of $U$ from $S$, can be done in time polynomial in $n$.
To finish, we show that the vertex cover $U$ has size at most $(1+\epsilon)K^*$,
where $K^*$ is the size of the optimal vertex cover in $G$.

By \lref[Lemma]{lem:cover-to-schedule}, $\JG$ has an order schedule of cost at most $10.5n+K^*+6$.
Since $G$ is cubic,  $K^* \geq m/3 = n/2$, so $10.5n+6 \le 11.5 n \le 23 K^*$.
Thus $S$ has cost at most

\begin{align*}
(1+\epsilon/24)(10.5n+K^*+6)
		~&=~ 10.5n+K^*+6 + (\epsilon/24) (10.5n+K^*+6)
		\\	
		~&\le~ 10.5n + K^\ast + 6 + (\epsilon/24) (23K^*+K^*)
		\\
		~&=~ 10.5n +  (1 + \epsilon) K^* + 6.
\end{align*}

Since all costs are integer, the cost of $S$ is in fact at most
$10.5n + K + 6$, where $K = \floor{ (1 + \epsilon) K^* }$.
Using this bound  and \lref[Lemma]{lem:schedule-to-cover},
the vertex cover $U$ has size at most $K \le (1+\epsilon) K^*$.
\end{proof}

Since $\JRPDEfour \in \APX$ (it has a constant-factor approximation algorithm),
the theorem implies that $\JRPDEfour$ is $\APX$-complete.

Of course, $\APX$-hardness implies that, unless $\P = \NP$,
there is no PTAS for $\JRPDEfour$:
that is, for some $\delta>0$,
there is no polynomial-time $(1+\delta)$-approximation algorithm for $\JRPDEfour$.



\section{Final Comments}
\label{sec: final comments}


The integrality gap for the standard $\JRPD$ LP relaxation is between 1.245 and 1.574.
We conjecture that neither bound is tight. Although we do not have a formal
proof, we believe that our
refined distribution for the tally game given here is optimal: it was
optimized under the assumption that it never generates more than two samples,
and allowing more than two samples, according to our calculations, can only
increase the value of $\statistic \distribution$. Thus
improving the upper bound will likely require a different approach.

There is a simple algorithm for $\JRPD$ that provides a $(1,2)$-approximation, meaning
that its warehouse order cost is not larger than that in the optimum, while
its retailer order cost is at most twice that in the optimum \cite{jrp-deadlines-nonner}.
One can combine that algorithm and the one here by choosing each algorithm with a
certain probability. This simple approach does not improve the approximation ratio,
but it may be possible to do so if, instead of using the algorithm presented here, 
one appropriately adjusts the probability distribution.

The computational complexity of general $\JRPD$, 
as a function of the maximum number $p$ of demand periods of each retailer,
is essentially resolved: for $p\ge 3$ the problem is
$\APX$-hard \cite{jrp-deadlines-nonner}, while for $p\le 2$ it can be solved in
polynomial time (for $p=1$ it can be solved with a greedy algorithm; 
for $p=2$ one can apply a dynamic programming algorithm
similar to that used in the proof of \lref[Lemma]{lem:3-bounded-optimal}).
For the case of equal-length demand periods, we showed that the problem remains
$\APX$-hard for $p\ge 4$. It would be nice to settle the case $p=3$,
which remains open.  We conjecture that this case is also $\NP$-complete.

Finally, we note that any LP-based algorithm for $\JRPD$ can be used as 
a building block for general $\JRP$ (with arbitrary waiting 
costs)~\cite{jrp-soda-2014}. The construction 
combines 
One-Sided Retailer Push and Two-Sided Retailer Push algorithms~\cite{jrp-owmr-levi-journal} 
with 
an appropriately crafted and scaled instance of $\JRPD$. By plugging our $1.574$-approximation to solve the 
$\JRPD$ instance, the algorithm of~\cite{jrp-soda-2014} yields 
a $1.791$-approximation for JRP.



\myparagraph{Acknowledgements.} 
We would like to thank {\L}ukasz Je{\.z}, Dorian Nogneng, Ji\v{r}\'{\i} Sgall,
and Grzegorz Stachowiak for stimulating discussions and useful comments.  We
are also grateful to anonymous reviewers of earlier versions of this manuscript
for pointing out several mistakes and suggestions for improving the
presentation.


\bibliographystyle{alpha}
\bibliography{references}

\newcommand{\etalchar}[1]{$^{#1}$}
\begin{thebibliography}{BMSV{\etalchar{+}}09}

\bibitem[AJR89]{jrp-arkin}
Esther Arkin, Dev Joneja, and Robin Roundy.
\newblock Computational complexity of uncapacitated multi-echelon production
  planning problems.
\newblock {\em Operations Research Letters}, 8(2):61--66, 1989.

\bibitem[AK00]{alimonti-apx-complete-cubic}
Paola Alimonti and Viggo Kann.
\newblock Some {APX}-completeness results for cubic graphs.
\newblock {\em Theoretical Computer Science}, 237(1--2):123--134, 2000.

\bibitem[BBC{\etalchar{+}}14]{jrp-soda-2014}
Marcin Bienkowski, Jaroslaw Byrka, Marek Chrobak, {\L}ukasz Je\.{z}, and
  Ji\v{r}\'{\i} Sgall.
\newblock Better approximation bounds for the joint replenishment problem.
\newblock In {\em Proc.~of the 25th ACM-SIAM Symp. on Discrete Algorithms
  (SODA)}, pages 42--54, 2014.

\bibitem[BKL{\etalchar{+}}08]{jrp-online-buchbinder}
Niv Buchbinder, Tracy Kimbrel, Retsef Levi, Konstantin Makarychev, and Maxim
  Sviridenko.
\newblock Online make-to-order joint replenishment model: Primal dual
  competitive algorithms.
\newblock In {\em Proc.~of the 19th ACM-SIAM Symp. on Discrete Algorithms
  (SODA)}, pages 952--961, 2008.

\bibitem[BKV12]{aggregation-bkv}
Carlos Brito, Elias Koutsoupias, and Shailesh Vaya.
\newblock Competitive analysis of organization networks or multicast
  acknowledgement: How much to wait?
\newblock {\em Algorithmica}, 64(4):584--605, 2012.

\bibitem[BMSV{\etalchar{+}}09]{packet-aggregation-becchetti}
Luca Becchetti, Alberto Marchetti-Spaccamela, Andrea Vitaletti, Peter Korteweg,
  Martin Skutella, and Leen Stougie.
\newblock Latency-constrained aggregation in sensor networks.
\newblock {\em ACM Transactions on Algorithms}, 6(1):13:1--13:20, 2009.

\bibitem[KNR02]{khanna-message-aggregation}
Sanjeev Khanna, Joseph Naor, and Danny Raz.
\newblock Control message aggregation in group communication protocols.
\newblock In {\em Proc.~of the 29th Int. Colloq. on Automata, Languages and
  Programming (ICALP)}, pages 135--146, 2002.

\bibitem[LRS05]{jrp-owmr-levi-soda}
Retsef Levi, Robin Roundy, and David~B. Shmoys.
\newblock A constant approximation algorithm for the one-warehouse
  multi-retailer problem.
\newblock In {\em Proc.~of the 16th ACM-SIAM Symp. on Discrete Algorithms
  (SODA)}, pages 365--374, 2005.

\bibitem[LRS06]{jrp-levi-2-approx}
Retsef Levi, Robin Roundy, and David~B. Shmoys.
\newblock Primal-dual algorithms for deterministic inventory problems.
\newblock {\em Mathematics of Operations Research}, 31(2):267--284, 2006.

\bibitem[LRSS08]{jrp-owmr-levi-journal}
Retsef Levi, Robin Roundy, David~B. Shmoys, and Maxim Sviridenko.
\newblock A constant approximation algorithm for the one-warehouse
  multiretailer problem.
\newblock {\em Management Science}, 54(4):763--776, 2008.

\bibitem[LS06]{jrp-owmr-levi-approx}
Retsef Levi and Maxim Sviridenko.
\newblock Improved approximation algorithm for the one-warehouse multi-retailer
  problem.
\newblock In {\em Proc.~of the 9th Int. Workshop on Approximation Algorithms
  for Combinatorial Optimization (APPROX)}, pages 188--199, 2006.

\bibitem[NS09]{jrp-deadlines-nonner}
Tim Nonner and Alexander Souza.
\newblock Approximating the joint replenishment problem with deadlines.
\newblock {\em Discrete Mathematics, Algorithms and Applications},
  1(2):153--174, 2009.

\end{thebibliography}


\appendix

\section{Wald's Lemma}
\label{sec: walds lemma}



Here is the variant of Wald's Lemma 
(also known as Wald's identity, and a consequence of standard ``optional stopping'' theorems)
that we use in \lref[Section]{sec: upper bounds}.
The proof is standard; we present it for completeness.


\begin{lemma}[Wald's Lemma]\label{lem:wald_equation}
Consider a random experiment that,
starting from a fixed start state $S_0$,
produces a random sequence of states $S_1,S_2,S_3,\ldots$
Let random index $T\in\{0,1,2,\ldots\}$ be a~stopping time for the sequence
(that is, for each positive integer $t$, the event ``$T \leq t$''
is determined by state $S_t$).
Let function $\phi:\{S_t\} \rightarrow \reals$ map the states to $\reals$.
Suppose that, for some fixed constants $\xi$ and $F$,

\smallskip
\noindent
{\rm (i)}  \((\forall t<T)~ \Exp[\phi(S_{t+1})~|~ S_t] \ge \phi(S_t)+\xi\),

\smallskip
\noindent
{\rm (ii)} either 
$(\forall t<T)~ \phi(S_{t+1})-\phi(S_t) \ge F$ in all outcomes,
or $(\forall t<T)~ \phi(S_{t+1})-\phi(S_t) \le F$ in all outcomes, and

\smallskip
\noindent
{\rm (iii)} $T$ has finite expectation.

\smallskip
\noindent
Then, $\Exp[\phi(S_T)]\, \ge \,\phi(S_0) + \xi\, \Exp[T]$.
\end{lemma}

\begin{proof}
For each $t\ge 0$, define random variable $\delta_t = \phi(S_{t+1})-\phi(S_{t})$.
By assumption (i), $\Exp[\delta_{t}~|~S_t] \ge \xi$ for $t<T$.
Since the event ``$T > t$'' is determined by $S_t$,
this implies that $\Exp[\delta_{t}~|~T>t] \ge \xi$. Then the inequality in the 
lemma can be derived as follows:
\begin{align*}
\textstyle
\Exp[\phi(S_T) - \phi(S_0)] ~=~ \Exp[ \sum_{t<T} \delta_t ]
	&=~
	\textstyle
	\sum_{\tau\ge 0} \Prob[T=\tau]\cdot\Exp[\sum_{t<\tau} \delta_t ~|~ T=\tau]
	\\
	&=~
	\textstyle
	\sum_{\tau \ge 0} \sum_{t<\tau} \Prob[T=\tau]\cdot\Exp[ \delta_t ~|~ T=\tau]
	\\
	&=~
	\textstyle
	\sum_{t \ge 0} \sum_{\tau>t} \Prob[T=\tau]\cdot\Exp[\delta_t ~|~ T=\tau]
	\\
	&=~
	\textstyle
	\sum_{t\ge 0} \Prob[T>t]\cdot\Exp[ \delta_t ~|~ T> t]
	\\
	&\ge~
	\textstyle
	\sum_{t\ge 0} \Prob[T>t]\cdot\xi
        &
	\\
	&=~
	\xi\, \Exp[T],
\end{align*}
Exchanging the order of summation in the third step above
does not change the value of the sum, 
because (by assumptions (ii) and (iii))
either the sum of all negative terms
is at least
$\sum_{\tau\ge 0} \sum_{t<\tau} \Pr[T=\tau] F
 ~=~ F  \sum_{\tau\ge 0} \tau \Pr[T=\tau] 
 ~=~ F\,\Exp[T]$,
which is finite, or (likewise) the sum of all positive terms is finite.
\end{proof}
 
Each application in \lref[Section]{sec: upper bounds} has $\xi \geq \statistic
\distribution > 0$ and $\phi(S_T) - \phi(S_0) \le U$ for some fixed~$U$. In
this case Wald's Lemma implies 
$\Exp[T] \leq U/\xi \leq U/\statistic \distribution$.


\end{document}